\documentclass[journal,comsoc]{IEEEtran}
\usepackage{cite}
\usepackage{amsmath,amssymb,amsfonts,texlogos}
\usepackage{algorithmic}
\usepackage{graphicx}
\usepackage{textcomp}
\usepackage{multicol,lipsum}
\usepackage{mathtools, cuted}
\usepackage{xcolor}
\def\BibTeX{{\rm B\kern-.05em{\sc i\kern-.025em b}\kern-.08em
    T\kern-.1667em\lower.7ex\hbox{E}\kern-.125emX}}

 \usepackage[linesnumbered,ruled]{algorithm2e}
\usepackage{amssymb,amsthm,textcomp}

\usepackage{url}
\usepackage[final]{pdfpages}
\usepackage{listings}
\usepackage{textcomp}
\usepackage{color} 
\usepackage{comment}

\usepackage{amsthm}
\usepackage{graphicx} 
\usepackage{mathtools}
\DeclareMathOperator*{\argmax}{arg\,max}
\DeclareMathOperator*{\argmin}{arg\,min}
 \usepackage{hyperref}
\usepackage{multirow}
\usepackage{color}
\usepackage{xcolor}

 \usepackage[protrusion=true,expansion=true]{microtype}
\usepackage[section]{placeins}
\usepackage{lipsum}
\usepackage{paralist}
\usepackage{subcaption}
\usepackage{wrapfig}
\theoremstyle{remark}
\usepackage{tabularx}
\newtheorem{definition}{Definition}[]
\newtheorem{theorem}{Theorem}[]

\newtheorem{proposition}{Proposition}[]
\newtheorem{lemma}{Lemma}[]
\newtheorem{remark}{Remark}[]

\newtheorem{corollary}{Corollary}[]
\usepackage[font=small]{caption}
\DeclarePairedDelimiter\floor{\lfloor}{\rfloor}
\makeatletter
\newcommand\footnoteref[1]{\protected@xdef\@thefnmark{\ref{#1}}\@footnotemark}
\makeatother

\DeclarePairedDelimiter\abs{\lvert}{\rvert}

\begin{document}
\title{{\color{black}Interference Avoidance Position Planning in Dual-hop and Multi-hop
UAV Relay Networks}}

\author{Seyyedali~Hosseinalipour,~\IEEEmembership{Student~Member,~IEEE,}
        Ali~Rahmati,~\IEEEmembership{Student~Member,~IEEE,}
        and~Huaiyu~Dai,~\IEEEmembership{Fellow,~IEEE}
\thanks{S. Hosseinalipour, A. Rahmati, and H. Dai are with the Department
of Electrical and Computer Engineering, North Carolina State University, Raleigh,
NC, USA e-mail: (\{shossei3,arahmat,hdai\}@ncsu.edu).
}
\thanks{
Part of this work was presented at the 2019 IEEE international conference on communication (ICC)~\cite{ourConf}.}
}
\maketitle
\begin{abstract}
We consider unmanned aerial vehicle (UAV)-assisted wireless communication employing UAVs as relays to increase the throughput between a pair of transmitter and receiver. We focus on developing effective methods to position the UAV(s) in the presence of interference in the environment, the existence of which makes the problem non-trivial and our methodology different from the current art. We study the optimal position planning, which aims to maximize the (average) {\color{black}signal-to-interference-ratio} (SIR) of the system, in the presence of: i) one major source of interference, ii) stochastic interference.  For each scenario, we first consider utilizing a single UAV in the dual-hop relay mode and determine its optimal position. Afterward, multiple UAVs in the multi-hop relay mode are considered, for which we investigate two novel problems concerned with determining the optimal number of required UAVs and developing an optimal distributed position alignment method. Subsequently, we propose a cost-effective method that simultaneously minimizes the number of UAVs and determines their optimal positions so as to guarantee a certain (average) SIR of the system. Alternatively, for a given number of UAVs, we develop a fully distributed placement algorithm along with its performance guarantee. Numerical simulations are provided to evaluate the performance of our proposed methods.  
\end{abstract}
\begin{IEEEkeywords}
UAV, relay networks, interference avoidance, position planning, interference mitigation. 
\end{IEEEkeywords}
\vspace{-2mm}
\section{Introduction}\label{sec:intro}
\noindent \IEEEPARstart{R}{ecently},  unmanned aerial vehicles (UAVs) have  been considered  as
a promising solution  for a variety of critical applications
such as environmental surveillance, public safety, disaster
relief, search and rescue, and purchase delivery~\cite{hayat2016survey}. Considering \textit{relaying} as one of the most elegant data transmission techniques in wireless communications~\cite{RelayImp1,RelayImp2,RelayImp3}, one of the recent applications of UAVs is utilizing them as relays in wireless networks~\cite{mag,7577063,mag2}.
Constructing a UAV communication network for
such  applications is a non-trivial task since there is no
regulatory and pre-allocated spectrum
band for the UAVs. As a result, this network usually coexists with other communication networks, e.g., cellular networks~\cite{saleem2015integration, rahmati2019energy}.
Thus, studying the problem of \textit{interference avoidance/mitigation} for the UAV communication network is critical, where the inherent mobility feature of the UAVs can be deployed as an interference evasion mechanism. This fact is the main motivation behind this work.

 In most of the related literature, the position planning for a single UAV, which is considered either as a gateway between a set of sensors and a ground node or as a relay node between a pair of transmitter and receiver, is developed~\cite{ladosz2016optimal,jiang2018power,nagubandi2018rasi,zhan2006wireless,ono2016wireless,wang2017improving,zeng2016throughput,zhang2018joint,8116613,chen2018local,faqir2018energy}. 
In~\cite{ladosz2016optimal}, the optimal position of a set of UAV relays is studied to improve  the  network connectivity  and  communication  performance  of  a  team  of ground   nodes/vehicles, where there is no communications among the UAVs themselves.
In~\cite{jiang2018power}, a UAV is employed as a mobile relay to ferry data between two disconnected ground nodes. This work aims to maximize the end-to-end throughput of the system by optimizing the source/relay power allocation and the UAV's trajectory.
In~\cite{nagubandi2018rasi}, UAV-assisted relay networks are studied in the context of cyber-physical systems, where a relay-based secret-key generation technique between two UAVs are proposed.
 In \cite{zhan2006wireless}, optimal deployment of a
UAV in a wireless relay communication
system is studied in order to improve the quality of communications between
two obstructed access points, while the symbol
error rate  is kept below a certain threshold.
 In \cite{ono2016wireless},  UAVs are utilized as  moving relays among the ground stations
with disconnected communication links in the event of disasters, where a variable-rate relaying approach is proposed to optimize the outage probability and information
rate.
In \cite{wang2017improving}, UAV-enabled mobile relaying in the context of the wiretap channel is proposed to facilitate secure
wireless communications, the goal of which is to maximize the secrecy rate of the system.
 In \cite{zeng2016throughput}, considering the usage of a UAV as relay between a pair of transmitter and receiver,  an end-to-end throughput
maximization problem  is formulated to optimize the
relay trajectory and the source/relay power allocations
in a finite time horizon.  
In \cite{zhang2018joint}, a UAV works as an
amplify-and-forward relay between a base station and a mobile device. The trajectory and the transmit power of the UAV and the transmit power of the mobile device are obtained aiming to minimize 
the outage probability of the system.
In \cite{8116613}, the placement of a UAV in
both static and mobile relaying 
 schemes is investigated to maximize the reliability of the network, for which the total power loss,
 the overall outage, and the overall bit error rate are used as
 reliability measures. Also, it is shown that 
that decode-and-forward relaying is better than amplify-and-forward relaying in terms of reliability.
In~\cite{chen2018local}, position planning of a UAV relay is studied to provide connectivity or a capacity boost for the ground users in a  dense  urban  area, where a nested segmented propagation  model  is  proposed  to  model  the  propagation from the UAV to the ground user that might be blocked by obstacles.
 In \cite{faqir2018energy}, the optimization of propulsion and transmission energies of a UAV relay is considered, where the problem is studied as an optimal control problem for energy minimization based on dynamic models for both transmission and mobility.
 Studying the UAV placement planning in the multi-hop relay communication context, in which multiple UAVs can be utilized between the transmitter and the receiver, is a new topic studied in~\cite{challita2017network,zhang2018trajectory,li2018placement,chen2018multiple}. The aim of these works is similar to the aforementioned literature; however, data transmission through multiple UAVs  makes their methodology different. Moreover, there are some similar works in the literature of sensor networks, among which the most relevant ones are~\cite{roh2010optimal, {chattopadhyay2017deploy}}. In \cite{roh2010optimal}, the two-dimensional (2-D) placement of relays is investigated aiming to increase the achievable transmission rate. In \cite{chattopadhyay2017deploy}, the \textit{impromptu} (as-you-go) placement of the relay nodes between a pair of source and sink node is addressed considering the distance between those nodes as a random variable, where the space is restricted to be (1-D). 
 
 Nevertheless, none of the aforementioned works consider the placement of UAV(s) in the presence of interference in the environment. This work can be broken down into two main parts. In the first part, we aim to go one step beyond the current literature and investigate the UAV-assisted wireless communication paradigm in the presence of a major source of interference (MSI), which refers to the source of interference with the dominant effect in the environment. Considering different interpretations for the MSI, e.g., a primary transmitter in UAV cognitive radio networks \cite{saleem2015integration,rahmati2019dynamic}, an eNodeB in UAV-assisted LTE-U/WiFi public safety networks \cite{athukoralage2016regret}, a malicious user in drone delivery application, or a base station in surveillance application, our paper can be adapted to multiple real-world scenarios. Given the intractability of direct analysis upon having multiple sources of interference in the network, we later show that the interference caused by multiple sources of interference with known locations can be modeled as the interference of a single hypothetical MSI, making our framework and analysis applicable to a wider range of applications. In the second part, we consider a distinct scenario, in which, due to the limited knowledge of the positions of the sources of interference or the time varying nature of the environment, we model the interference as a stochastic phenomenon. For each part, we study the optimal placement planning upon having a single UAV, i.e., dual-hop single link scheme, and multiple UAVs, i.e., multi-hop single link scheme, acting as relays between the transmitter and the receiver. The existence of interference renders our methodology different compared to the current literature; however, the previously derived results can be considered as especial cases in our model by assuming that the MSI is located too far away or it possesses an insignificant transmitting power. Hence, the methodology proposed in this work can motivate multiple follow up works revisiting the previously studied problems  considering the presence of interference in their models. Moreover, compared with the relevant literature on multi-hop UAV-assisted relay communication, e.g.,~\cite{challita2017network,zhang2018trajectory,li2018placement,chen2018multiple}, mostly focused on obtaining the optimal location/trajectory of the UAVs, in addition to incorporating the interference into our model, we introduce and investigate two new problems: i) determining the minimum required number of UAVs and their locations so as to satisfy a desired (average) {\color{black}signal-to-interference-ratio (SIR)}, or equivalently data rate, of the system; ii) developing a distributed placement algorithm, which requires message passing only among adjacent UAVs to maximize the (average) SIR of the system.
 \subsection{Contributions}

 1) We investigate the problem of optimal UAV position planning considering the effect of interference in the environment in the decode-and-froward relay communication context for both the dual-hop and the multi-hop relay settings. We pursue the problem considering i) the existence of an MSI in the network, and ii) stochastic interference. Moreover, we propose and investigate two novel problems in the multi-hop relay setting: i) determining the minimum required number of UAVs and their optimal positions, and ii) developing a distributed position alignment algorithm.
 
 2) Considering a single UAV and an MSI, we develop a theoretical approach to identify the optimal position of the UAV to maximize the SIR of the system in the dual-hop setting.  We also address the position planning for a single UAV upon having stochastic interference.
 
 3)  In the multi-hop relay context, considering the existence of an MSI, we develop a theoretical framework that simultaneously determines the minimum required number of UAVs and their optimal positions so as to satisfy a predetermined/desired SIR of the system. We also develop a similar framework considering the stochastic interference in the environment and investigate the optimally of our approach upon having independent and identically distributed (i.i.d.) and non-i.i.d interference along the horizontal axis.
 
 4)  In the multi-hop relay context, considering the existence of an MSI and given the number of UAVs, we propose an optimal distributed algorithm that achieves the maximum attainable SIR of the system, which only requires message exchange among the adjacent UAVs. We also propose a distributed position planning considering stochastic interference and investigate its optimality upon having i.i.d. and non-i.i.d interference along the horizontal axis.
 
 \section{Preliminaries}
\noindent We consider data transmission between a pair of transmitter (Tx) and receiver (Rx) co-existing with a major source of interference (MSI). We consider a \textit{left-handed coordination system} $(x,y,h)$, where the Tx, the Rx, and the MSI are assumed to be on the ground plane defined as $h=0$. The locations of the Tx, the Rx, and the MSI are assumed to be $(0,0,0)$, $(D,0,0)$, and $(X_{_{\textrm{MSI}}},Y_{_{\textrm{MSI}}},0)$, respectively. {\color{black} In practice, the location of the MSI can be estimated using jammer localization techniques (see \cite{jammerLoc} and references therein).} We assume $0\leq X_{_{\textrm{MSI}}}\leq D$ for simplicity, which can be readily generalized with minor modification. The transmission powers of the Tx, the UAV, and the MSI are denoted by $p_t$, $p_u$, and $p_{_{\textrm{MSI}}}$, respectively. To improve the transmission data rate, it is desired to place a UAV or a set of UAVs, each of which acting as a relay, between the Tx and the Rx. To have tractable solutions, we assume that the UAVs are placed at $y=0$ plane. While such a constraint impose certain limitations to our study, it allows us to obtain some first analytical results that provide insightful guidance for practical design in general and also are meaningful for some specific application scenarios. Also, considering legal regulations, we confine the altitude of the UAVs to $h\in[h_{min},h_{max}]$. 

 \begin{table*}[hbt!]
\centering
   \caption{Major notations.}\label{table:1}
 \begin{tabular}{|c|c|} 
 \hline
 $X_{_{\textrm{MSI}}}$& The horizontal position of the MSI in $h=0$ plane \\ 
\hline
$Y_{_{\textrm{MSI}}}$ & The vertical position of the MSI in $h=0$ plane  \\ 
\hline
 $h$& Altitude of the UAV\\ \hline
 $D$& The distance between the Tx and the Rx\\ \hline
  $p_t$& Transmission power of the Tx \\\hline
    $p_u$& Transmission power of the UAVs \\\hline
    $p_{_{\textrm{MSI}}}$& Transmission power of the MSI \\\hline
    $\textrm{SIR}_S$ & SIR of the system\\\hline
    $d_i$& The distance between node $i$ and $i-1$ in the multi-hop setting ($1\leq i\leq N+1$)\\\hline
    $N$& The number of UAVs in the multi-hop setting \\\hline
    $I_x$& A random variable denoting the power of interference at horizontal position $x$\\\hline
     $M_{I_{x}}$& The moment generating function of $I_x$\\\hline
    $\Bar{I}_{x}$ & A random variable denoting the normalized power of interference at horizontal position $x$\\\hline
 \end{tabular}
\end{table*}

 We consider the line-of-sight (LoS) and the non-line-of-sight (NLoS) channel models, for which the path-loss is given by:
\begin{equation}
    L^{\textrm{LoS}}_{i,j}= \mu_{_{\textrm{LoS}}} d_{i,j}^{\alpha},\;\;L^{\textrm{NLoS}}_{i,j}=\mu_{_{\textrm{NLoS}}} d_{i,j}^{\alpha},
\end{equation}
where $\mu_{_{\textrm{LoS}}}\triangleq C_{_{\textrm{LoS}}}\left(4\pi f_c/c\right)^\alpha$, $\mu_{_{\textrm{NLoS}}}\triangleq C_{_{\textrm{NLoS}}}\left(4\pi f_c/c\right)^\alpha$, $C_{_{\textrm{LoS}}}$ ($C_{_{\textrm{NLoS}}}$) is the excessive path loss factor incurred by shadowing, scattering, etc., in the LoS (NLoS) link, $f_c$ is the carrier frequency, $c$ is the speed of light, $\alpha=2$ is the path-loss exponent\footnote{The LOS model is used for the air-to-air channel between the UAVs, for which $\alpha=2$ is a well-known choice. In some scenarios, the value of $\alpha$ for the NLOS link is assumed to be greater than $2$. This leads to straightforward modifications in the derived results.}, and $d_{i,j}$ is the Euclidean distance between node $i$ and node $j$. The link between two UAVs (air-to-air) is modeled using the LoS model, while the link between the MSI and the Rx (ground-to-ground) is modeled based on the NLoS model. To model the link between a UAV and the Rx/Tx/MSI (air-to-ground and ground-to-air) either the LoS or the NLoS model~\cite{chen2018multiple,zhang2018trajectory} or a weighted average between the LoS model and the NLoS model~\cite{moza:internet,channel2,channel3} can be used. In this paper, we consider a general case and denote the path loss between a UAV $i$ and node $j$ located on the ground by $\eta_{_{\textrm{NLoS}}}d_{ij}^{2}$. We assume that $\eta_{_{\textrm{NLoS}}}$ is constant in the range $h\in[h_{min},h_{max}]$, and thus $\eta_{_{\textrm{NLoS}}}\triangleq g(\mu_{_{\textrm{LoS}}},\mu_{_{\textrm{NLoS}}},h_{min},h_{max})$, where $g$ is a function. Further discussions on obtaining the $g$ in different environments can be found in~\cite{moza:internet,channel2,channel3}.  Due to
the geographical limitations, direct communication between
the Tx and Rx is not considered, which is a valid assumption especially when the Tx and the Rx are far away or there are obstacles between them~\cite{chen2018multiple,zhang2018trajectory}.

\section{Position Planning for a Single UAV Considering an MSI}\label{sec:singleUAV}
\noindent Let $\textrm{SIR}_1$, $\textrm{SIR}_2$ denote the SIR at the UAV located at $(x,0,h)$ and the SIR at the Rx, respectively (see Fig.~\ref{fig:single}), which are given by:
\begin{equation*}
\hspace{-5mm}
\resizebox{0.9\hsize}{!}{$
      \textrm{SIR}_1(x,h)\hspace{-1mm}=\hspace{-1mm}\frac{p_t\hspace{-.14mm}/\hspace{-.15mm}(\hspace{-.15mm}\eta_{_{\textrm{NLoS}}} d_{_{\textrm{UAV},\textrm{Tx}}}^{2}\hspace{-.15mm})}{p_{_{\textrm{MSI}}}\hspace{-.15mm}/\hspace{-.15mm}(\hspace{-.15mm}\eta_{_{\textrm{NLoS}}} d_{_{\textrm{UAV},\textrm{MSI}}}^{2}\hspace{-.15mm})}\hspace{-1mm}=\hspace{-1mm}\frac{p_t \hspace{-.12mm} \left(\hspace{-.15mm}\left(\hspace{-.15mm}x-X_{_{\textrm{MSI}}}\hspace{-.15mm}\right)^2\hspace{-.15mm}+\hspace{-.14mm}Y_{_{\textrm{MSI}}}^2\hspace{-.15mm}+\hspace{-.12mm}h^2\hspace{-.15mm}\right)}{p_{_{\textrm{MSI}}}\left(x^2+h^2\right)},$
      }
\end{equation*}
\begin{equation}\label{eq:SIRexp}
\hspace{-0.0mm}
\resizebox{0.9\hsize}{!}{$
      \textrm{SIR}_2\hspace{-.1mm}(\hspace{-.15mm}x,h\hspace{-.15mm})\hspace{-1mm}= \hspace{-.199mm}\frac{\hspace{-.19mm}p_u\hspace{-.15mm}/\hspace{-.15mm}(\hspace{-.15mm}\eta_{_{\textrm{NLoS}}} \hspace{-.15mm}d_{_{\textrm{UAV},\textrm{Rx}}}^{2}\hspace{-.15mm})}{\hspace{-0.17mm}p_{_{\textrm{MSI}}}\hspace{-.15mm}/\hspace{-.15mm}(\hspace{-.15mm}\mu_{_{\textrm{NLoS}}} \hspace{-.15mm}d_{_{\textrm{Rx},\textrm{MSI}}}^{2}\hspace{-.15mm})}\hspace{-1mm}=
      \hspace{-1mm} \frac{\hspace{-.15mm}p_u\hspace{-.1mm} \left(\hspace{-.15mm}Y^2_{_{\textrm{MSI}}}\hspace{-.15mm}+\hspace{-.0mm}(\hspace{-.13mm}D-X_{_{\textrm{MSI}}}\hspace{-.13mm})^2\hspace{-.13mm}\right)}{\hspace{-0.1mm}p_{_{\textrm{MSI}}}\hspace{-0.15mm}\left(\left(\hspace{-0.14mm}D\hspace{-0.12mm}-\hspace{-0.12mm}x\hspace{-0.13mm}\right)\hspace{-.12mm}^2\hspace{-0.12mm}+\hspace{-0.2mm}h\hspace{-.2mm}^2\right)\hspace{-0.0mm} \left(\frac{\hspace{-0.15mm}\eta_{_{\textrm{NLoS}}}}{\hspace{-0.15mm}\mu_{_{\textrm{NLoS}}}}\hspace{-.17mm}\right)}.
      $}
\end{equation}
Considering the conventional \textit{decode-and-forward} relay mode, the SIR of the system $\textrm{SIR}_S$ is given by~\cite{chen2018multiple}:
 \begin{equation}\label{eq:SIR_s}
     \textrm{SIR}_S(x,h)=\min \big\{\textrm{SIR}_1(x,h),\textrm{SIR}_2(x,h)\big\} \;\;\forall x,h.
\end{equation}
Assuming equal bandwidths for both links, maximizing the data rate between the Tx and the Rx is equivalent to maximizing the $\textrm{SIR}_S$ by tuning the location of the UAV described as:
\begin{equation}\label{mainOpt}
    (x^*,h^*) = \argmax_{x\in [0,D], h\in [h_{min},h_{max}]} \textrm{SIR}_S(x,h).
\end{equation}
The presence of an MSI renders our approach different from most of the works mentioned in Section~\ref{sec:intro} mainly due to its effect on the SIR expressions making them non-convex with respect to (w.r.t) the position of the UAV(s), which leads to the inapplicability of the conventional optimization techniques. In this work, we exploit \textit{geometry} and \textit{functional analysis} to obtain the subsequent derivations. In the following, we propose two lemmas, which are later used to derive the main results.
\begin{figure}[t]
\includegraphics[width=8.9cm,height=2.3cm]{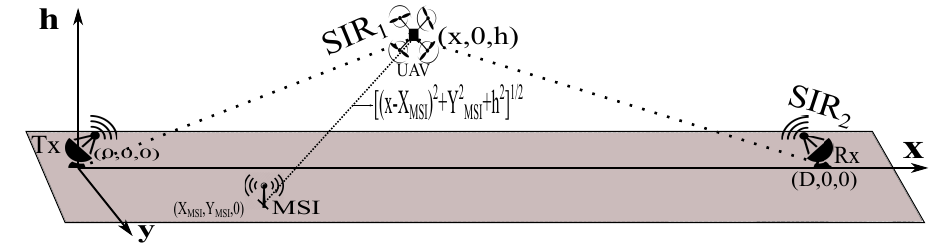}
		\caption{A single UAV acting as a relay between a pair of Tx and Rx coexisting with an MSI (dual-hop single link).}
		 \label{fig:single}
		 \end{figure}
\begin{definition}
In geometry, a \textit{locus} is the set of all points satisfying the same conditions or possessing the same properties.
\end{definition}
\begin{lemma}\label{th:main}
The locus of the points satisfying $\textrm{SIR}_1(x,h)=\textrm{SIR}_2(x,h)$ is given by the following expression\footnote{In this work, $+$ and $-$ superscripts always denote the larger and the smaller solution, respectively.}:
\begin{equation}\label{eq:locus}
    {h^{\pm}}=\sqrt{\Lambda^{\pm}(x)},
\end{equation}
with $\Lambda^{\pm}(x)\triangleq \big[-B(x)\pm \sqrt{B^2(x)-4A(x)C(x)}\big]/\left(2A(x)\right)$, where $A(x)$, $B(x)$, and $C(x)$ are given by~\eqref{eq:lambda}.
\begin{table*}[t]
\begin{minipage}{0.99\textwidth}
\begin{equation}\label{eq:lambda}
\begin{aligned}
     &A(x)=p_t, B(x)= p_t\left(X_{_{\textrm{MSI}}}-x\right)^2+p_t(D-x)^2-p_u       \left(\frac{\mu_{_{\textrm{NLoS}}}}{\eta_{_{\textrm{NLoS}}}}\right)\left(D-X_{_{\textrm{MSI}}}\right)^2+Y_{_{\textrm{MSI}}}^2\left(p_t-p_u       \left(\frac{\mu_{_{\textrm{NLoS}}}}{\eta_{_{\textrm{NLoS}}}}\right)\right), \\
     &C(x)= p_t\big[(D-x)^2\left(\left(X_{_{\textrm{MSI}}}-x\right)^2+Y_{_{\textrm{MSI}}}^2\right) \big]- p_u       \left(\frac{\mu_{_{\textrm{NLoS}}}}{\eta_{_{\textrm{NLoS}}}}\right)\Big[x^2\left(Y_{_{\textrm{MSI}}}^2+ (D-X_{_{\textrm{MSI}}})^2\right) \Big]
     \end{aligned}
\end{equation}
\hrulefill
\end{minipage}
\end{table*}
\end{lemma}
\begin{proof}
The proof can be carried out using algebraic manipulations, which is omitted due to the limited space.
\end{proof}
\begin{lemma}\label{lemma:lem1}
For $\textrm{SIR}_1$, the \textit{stationary point}~\cite{ref:stat} with respect to $x$, $\Psi^x$, is given by: 
\begin{equation}
    \begin{aligned}
  &\Psi^x= \frac{Y_{_{\textrm{MSI}}}^2+X_{_{\textrm{MSI}}}^2+\sqrt{(Y_{_{\textrm{MSI}}}^2+X_{_{\textrm{MSI}}}^2)^2+4X_{_{\textrm{MSI}}}^2h^2}}{2X_{_{\textrm{MSI}}}}.
\end{aligned}
\end{equation}
Also, $\textrm{SIR}_1$ has no stationary point with respect to $h$ when $h\in (h_{min},h_{max})$. With $\Psi^h\triangleq \frac{Y_{_{\textrm{MSI}}}^2+X_{_{\textrm{MSI}}}^2}{2X_{_{\textrm{MSI}}}}$, we have
\begin{equation}
    \begin{cases}
          \frac{\partial \textrm{SIR}_1(x,h)}{\partial x}\geq0 \;\textrm{if}\; x\geq \Psi^x,\\
          
          \frac{\partial \textrm{SIR}_1(x,h)}{\partial x}<0 \;\; \textrm{O.W.},
    \end{cases}
\;\;
    \begin{cases}
          \frac{\partial \textrm{SIR}_1(x,h)}{\partial h}\geq0\;\textrm{if}\;x\geq \Psi^h,\\
          \frac{\partial \textrm{SIR}_1(x,h)}{\partial h}<0\;\; \textrm{O.W}.
    \end{cases}
\end{equation}
Moreover,
\begin{equation}
   \hspace{-9mm}\resizebox{0.99\hsize}{!}{$ \displaystyle\max_{\hspace{9mm}x\in [0,D], h\in [h_{min},h_{max}]} \hspace{-6mm} \textrm{SIR}_1(x,h)=\frac{p_t(X_{_{\textrm{MSI}}}^2+Y_{_{\textrm{MSI}}}^2+h_{min}^2)}{p_{_{\textrm{MSI}}} h_{min}^2}.$}
\end{equation}
On the other hand, $\textrm{SIR}_2$ has no stationary point when $x\in (0,D), h\in (h_{min},h_{max})$ and
\begin{align}
         & \hspace{-.5mm}\frac{\partial \textrm{SIR}_2(x,h)}{\partial h}\leq0, \;\;
          \frac{\partial \textrm{SIR}_2(x,h)}{\partial x}\geq 0,\nonumber\\
          &\forall x\in [0,D], h\in [h_{min},h_{max}],
          \end{align}
          and
          \begin{equation}
    \hspace{-10mm}\max_{\hspace{10mm} x\in [0,D], h\in [h_{min},h_{max}]}\hspace{-14mm} \textrm{SIR}_2(x,h) =\hspace{-0.5mm}\frac{p_u\mu_{_{\textrm{NLoS}}}(Y_{_{\textrm{MSI}}}^2+(D-X_{_{\textrm{MSI}}})^2)}{p_{_{\textrm{MSI}}}\eta_{_{\textrm{NLoS}}}h_{min}^2}.
\end{equation}
\end{lemma}
\begin{proof}
The proof can be carried out by taking the following steps: (i) Analysis of $\frac{\partial \textrm{SIR}_1}{\partial x},\frac{\partial \textrm{SIR}_1}{\partial h},\frac{\partial \textrm{SIR}_2}{\partial x},\frac{\partial \textrm{SIR}_2}{\partial h}$ to obtain the stationary points. (ii) Examining the signs of $\frac{\partial^2 \textrm{SIR}_1}{\partial x^2} \frac{\partial^2 \textrm{SIR}_1}{\partial h^2}-\left(\frac{\partial^2 \textrm{SIR}_1}{\partial x \partial h}\right)^2$, $\frac{\partial^2 \textrm{SIR}_1}{\partial x^2}$, and $\frac{\partial^2 \textrm{SIR}_1}{\partial h^2}$ at the stationary points. (iii) Inspecting the behavior of the SIR expressions at the boundary points. 
\end{proof}

In practice, one of the following scenarios may occur: (i) The UAV position is vertically fixed and horizontally adjustable \cite{8449221,chen}. This may arise in urban applications, in which there is a desired altitude for the UAVs to avoid collision with other flying objects. (ii) The UAV position is horizontally fixed and vertically adjustable.  This happens specially in the surveillance and information gathering applications, in which the position of the UAV is fixed in the desired horizontal position and only the altitude can be tuned \cite{8531711}. (iii) The UAV position is neither vertically nor horizontally fixed, which is practical in non-urban areas with a few flying objects. In the following, we tackle these scenarios in order. 
Henceforth, whenever we refer to the roots of an equation or the points in the locus, the feasible space is confined to $x\in[0,D]$ and $h\in[h_{min},h_{max}]$.

\subsubsection{Finding the optimal horizontal position $x^*$ of the UAV for a given altitude $h=\hat{h}$}
In this case, we first analyze the result of Lemma~\ref{th:main} using Lemma~\ref{lemma:lem1} in the following corollary, based on which the optimal placement of the UAV is derived in Theorem~\ref{th:givenh}. 
\begin{table*}[t]
\begin{minipage}{0.99\textwidth}
\begin{equation}\label{eq:quartic}
\resizebox{0.95\hsize}{!}{$
p_t\left(x-X_{_{\textrm{MSI}}}\right)^2\left((D-x)^2+\hat{h}^2\right)+p_t( Y_{_{\textrm{MSI}}}^2+\hat{h}^2)(D-x)^2-p_u\left(\frac{\mu_{_{\textrm{NLoS}}}}{\eta_{_{\textrm{NLoS}}}}\right) x^2(Y_{_{\textrm{MSI}}}^2+(D-X_{_{\textrm{MSI}}})^2)+p_t \hat{h}^2(Y_{_{\textrm{MSI}}}^2+\hat{h}^2)-p_u\left(\frac{\mu_{_{\textrm{NLoS}}}}{\eta_{_{\textrm{NLoS}}}}\right) \hat{h}^2(Y_{_{\textrm{MSI}}}^2+(D-X_{_{\textrm{MSI}}})^2)=0
   $ }
\end{equation}
\end{minipage}
\end{table*}
\begin{corollary}\label{cor:1}
Given a fixed altitude $h=\hat{h}$, the horizontal positions satisfying \eqref{eq:locus} can be obtained by solving the \textit{quartic} equation given in~\eqref{eq:quartic}, where the characteristic of this equation considering $x\in [0,D]$ is described as follows:

\textbf{Case 1)} $\textrm{SIR}_1(0,\hat{h})< \textrm{SIR}_2(0,\hat{h})$: In this case, the quartic equation has no solution. With some algebraic manipulations, this case can be represented as the following constraint:
  \begin{equation}\label{con1}
     \frac{p_t}{p_u} < \frac{\mu_{_{\textrm{NLoS}}}(Y_{_{\textrm{MSI}}}^2+(D-X_{_{\textrm{MSI}}})^2)\hat{h}^2}{\eta_{_{\textrm{NLoS}}}(X_{_{\textrm{MSI}}}^2+Y_{_{\textrm{MSI}}}^2+\hat{h}^2)(D^2+\hat{h}^2)} \triangleq C_1.
  \end{equation}
Therefore, the necessary condition to have at least a feasible solution for~\eqref{eq:quartic} is $\textrm{SIR}_1(0,\hat{h})\geq \textrm{SIR}_2(0,\hat{h})$, which can be represented as $p_t/p_u\geq C_1$.

 \textbf{Case 2)} $p_t/p_u\geq C_1$ and $\Psi^x\geq D$ and $\textrm{SIR}_2(D,\hat{h}) \geq \textrm{SIR}_1(D,\hat{h})$: In this case, the quartic equation has one solution $x_{sol}$, which can be numerically obtained. This case can be represented by the following conditions: 
  \begin{equation}
     \begin{aligned}
          &\hspace{-0mm}\frac{Y_{_{\textrm{MSI}}}^2\hspace{-.5mm}+X_{_{\textrm{MSI}}}^2\hspace{-.5mm}+\hspace{-0mm}\sqrt{(Y_{_{\textrm{MSI}}}^2+X_{_{\textrm{MSI}}}^2)^2+4X_{_{\textrm{MSI}}}^2\hat{h}^2}}{2X_{_{\textrm{MSI}}}}\geq D,\\
          &\hspace{-0mm}C_1\leq \frac{p_t}{p_u} \leq C_2,\\
          &\hspace{-0mm}C_2\triangleq \frac{\mu_{_{\textrm{NLoS}}}(Y_{_{\textrm{MSI}}}^2+(D-X_{_{\textrm{MSI}}})^2)(D^2+\hat{h}^2)}{\eta_{_{\textrm{NLoS}}}\hat{h}^2((D-X_{_{\textrm{MSI}}})^2+Y_{_{\textrm{MSI}}}^2+\hat{h}^2)} .
          \end{aligned}
     \end{equation}
   
    \textbf{ Case 3)} $p_t/p_u\geq C_1$ and $\Psi^x\geq D$ and $\textrm{SIR}_2(D,\hat{h}) < \textrm{SIR}_1(D,\hat{h})$: In this case, the quartic equation has no solution.  This case can be represented by the following conditions:
    \begin{equation}
     \begin{aligned}
          &\frac{Y_{_{\textrm{MSI}}}^2\hspace{-.5mm}+\hspace{-.5mm}X_{_{\textrm{MSI}}}^2\hspace{-.5mm}+\hspace{-.5mm}\sqrt{(Y_{_{\textrm{MSI}}}^2+X_{_{\textrm{MSI}}}^2)^2+4X_{_{\textrm{MSI}}}^2\hat{h}^2}}{2X_{_{\textrm{MSI}}}}>D,\\
          & \frac{p_t}{p_u}> \max \big\{C_1,C_2\big\}.
          \end{aligned}
     \end{equation}
     
  \textbf{  Case 4) }$p_t/p_u\geq C_1$ and $\Psi^x<D$ and $\textrm{SIR}_1(\Psi^x,\hat{h})\leq \textrm{SIR}_2(\Psi^x,\hat{h})$: In this case, the quartic equation has at least a feasible solution. This condition can be represented as follows:
  \begin{equation}
     \begin{aligned}
          &\frac{Y_{_{\textrm{MSI}}}^2+X_{_{\textrm{MSI}}}^2+\sqrt{(Y_{_{\textrm{MSI}}}^2+X_{_{\textrm{MSI}}}^2)^2+4X_{_{\textrm{MSI}}}^2\hat{h}^2}}{2X_{_{\textrm{MSI}}}}<D,\\
          &C_1\leq \frac{p_t}{p_u} \leq C_3,
          \end{aligned}
     \end{equation}
     where $C_3$ is defined in~\eqref{C_3}.
     \begin{table*}
        \begin{minipage}{0.99\textwidth}
\begin{equation}\label{C_3}
    C_3\triangleq \frac{\mu_{_{\textrm{NLoS}}}\left(\left(\frac{Y_{_{\textrm{MSI}}}^2+X_{_{\textrm{MSI}}}^2+\sqrt{(Y_{_{\textrm{MSI}}}^2+X_{_{\textrm{MSI}}}^2)^2+4X_{_{\textrm{MSI}}}^2h^2}}{2X_{_{\textrm{MSI}}}} \right)^2+h^2\right)(Y_{_{\textrm{MSI}}}^2+(D-X_{_{\textrm{MSI}}})^2) }{\eta_{_{\textrm{NLoS}}}\left(\left(\left(\frac{Y_{_{\textrm{MSI}}}^2+X_{_{\textrm{MSI}}}^2+\sqrt{(Y_{_{\textrm{MSI}}}^2+X_{_{\textrm{MSI}}}^2)^2+4X_{_{\textrm{MSI}}}^2h^2}}{2X_{_{\textrm{MSI}}}} \right)-X_{_{\textrm{MSI}}}\right)^2+Y_{_{\textrm{MSI}}}^2+h^2\right)\left(\left(D-\frac{Y_{_{\textrm{MSI}}}^2+X_{_{\textrm{MSI}}}^2+\sqrt{(Y_{_{\textrm{MSI}}}^2+X_{_{\textrm{MSI}}}^2)^2+4X_{_{\textrm{MSI}}}^2h^2}}{2X_{_{\textrm{MSI}}}}\right)^2+h^2\right)}
\end{equation}
\begin{minipage}{1.0065\textwidth}
  \begin{tabularx}{1\textwidth}{Xp{6cm}X}
  \begin{equation}\label{eq:case1}
  \hspace{-15.0mm}
\resizebox{0.90\hsize}{!}{$
        \textit{p}_u/p_t \leq \max \Bigg\{\frac{\eta_{_{\textrm{NLoS}}}\left(\left(D-X_{_{\textrm{MSI}}}\right)^2+Y_{_{\textrm{MSI}}}^2+h_{max}^2 \right)h_{min}^2}{\mu_{_{\textrm{NLoS}}}\left( Y_{_{\textrm{MSI}}}^2+\left(D-X_{_{\textrm{MSI}}}\right)^2\right)\left(D^2+h_{max}^2\right)},\frac{\eta_{_{\textrm{NLoS}}}\left(\left(D-X_{_{\textrm{MSI}}}\right)^2+Y_{_{\textrm{MSI}}}^2+h_{min}^2 \right)h_{min}^2}{\mu_{_{\textrm{NLoS}}}\left( Y_{_{\textrm{MSI}}}^2+\left(D-X_{_{\textrm{MSI}}}\right)^2\right)\left(D^2+h_{min}^2\right)} \Bigg\}\hspace{-2.400em}$}
    \end{equation}
      &
  \begin{equation}\label{eq-b}
    \hspace{-5.55mm}
    \resizebox{0.82\hsize}{!}{ $
   {\textit{p}_t}/{p_u}\leq \frac{\mu_{_{\textrm{NLoS}}} h^2_{min}\left(Y_{_{\textrm{MSI}}}^2+\left(D-X_{_{\textrm{MSI}}}\right)^2 \right)}{\hspace{-2.59mm}\eta_{_{\textrm{NLoS}}}\left(X_{_{\textrm{MSI}}}^2+Y_{_{\textrm{MSI}}}^2+h^2_{min} \right)\left(D^2+h^2_{min} \right)}
  \hspace{-2.5mm} $} 
  \end{equation}
  \end{tabularx}
\end{minipage}
\end{minipage}
\end{table*}

      \textbf{   Case 5)} $p_t/p_u\geq C_1$ and $\Psi^x<D$ and $\textrm{SIR}_1(\Psi^x,\hat{h})> \textrm{SIR}_2(\Psi^x,\hat{h})$: In this case, the quartic equation may or may not have a feasible solution. This condition can be expressed as follows:
     \begin{equation}\label{conL}
         \begin{aligned}
            &\frac{Y_{_{\textrm{MSI}}}^2+X_{_{\textrm{MSI}}}^2+\sqrt{(Y_{_{\textrm{MSI}}}^2+X_{_{\textrm{MSI}}}^2)^2+4X_{_{\textrm{MSI}}}^2\hat{h}^2}}{2X_{_{\textrm{MSI}}}}<D,\\
            &\frac{p_t}{p_u}\geq \max\{C_1,C_3\}.
         \end{aligned}
     \end{equation}
\end{corollary}
\begin{proof}
  For a fixed altitude, according to Lemma~\ref{lemma:lem1}:  (i) $\textrm{SIR}_2$ is a monotone increasing function w.r.t $x$, and
    (ii) depending on the value of the stationary points of $\textrm{SIR}_1$, $\textrm{SIR}_1$ is a monotone decreasing function w.r.t $x$ in the interval $x\in [0,\Psi^x)$ and a non-decreasing function w.r.t $x$ in the interval $x\in[\Psi^x,D]$. This corollary is a result of these two facts combined with the usage of functional analysis. 
\end{proof}

In the following theorem, we use the results of Corollary~\ref{cor:1} to determine the optimal position of the UAV. However, the above corollary also provides a practical guide to design the $p_t$ and $p_u$ w.r.t the position of the MSI, which can be obtained through calculation of $C_1$ through $C_3$, and the conditions given on the ratio of these two variables in~\eqref{con1}-\eqref{conL}. Similarly, it discloses useful guides for the malicious user to effectively place the MSI. Nonetheless, we leave these interpretations as future work since they are not the focus of this paper. 
\begin{theorem}\label{th:givenh}
Given a fixed altitude $h=\hat{h}$, the optimal horizontal position of the UAV $x^*$ for the cases defined in Corollary~\ref{cor:1} is as follows: In case 1, $x^*=0$. In case 2, $x^*=x_{sol}$. In case 3, $x^*=D$. In case 4, let $x_{sol}$ denote the smallest solution of the quartic equation~\eqref{eq:quartic}, if $\textrm{SIR}_1(x_{sol},\hat{h})\geq \textrm{SIR}_1(D,\hat{h})$ then $x^*=x_{sol}$; otherwise, $x^*=D$.
In case 5, $x^*=D$.
\end{theorem}
\begin{proof}
The proof is an immediate result of Corollary~\ref{cor:1} considering the behaviors of the SIR expressions given in Lemma~\ref{lemma:lem1}.
\end{proof}

\subsubsection{Finding the optimal vertical position $h^*$ of the UAV for a given horizontal position $x=\hat{x}$}
In this case, the vertical positions (altitudes) satisfying~\eqref{eq:locus} can be easily derived since $\Lambda^{\pm}(x)$ on the right hand side of the equation is known. Using Lemma~\ref{lemma:lem1}, we obtain the following theorem to identify the optimal position of the UAV.
\begin{theorem}\label{th:givenX}
Given a fixed horizontal position $x=\hat{x}$, the optimal altitude $h^*$ of the UAV is given by:

\textbf{Case 1)} $\hat{x}\leq \Psi^h$: $h^*=h_{min}$. 
 
\textbf{Case 2)} $\hat{x}> \Psi^h$ and~\eqref{eq:locus} has a feasible solution (either $h^+$ or $h^-$ belong to $[h_{min},h_{max}]$): $h^*$ is the same as the feasible solution of~\eqref{eq:locus}.
 
\textbf{Case 3)} $\hat{x}>\Psi^h$ and~\eqref{eq:locus} has no feasible solution: $h^*$ can be derived by solely inspecting the boundary positions:
     \begin{equation}
         h^*=\argmax_{h\in \{h_{min},h_{max}\}} \textrm{SIR}_S(\hat{x},h).
     \end{equation}
\end{theorem}
\begin{proof}
The proof is an immediate result of studying the behaviors of the SIR expressions given in Lemma~\ref{lemma:lem1}.
\end{proof}
\subsubsection{Finding the optimal position when both $h$ and $x$ of the UAV are adjustable}
In the previous scenarios, the locus defined in~\eqref{eq:locus} reduces to an equation since one variable (either $h$ or $x$) is given, which is not the case here. In this case, the optimal position of the UAV is identified in the following theorem. 

\begin{theorem}\label{th:var}
Let $\Lambda$ denote the set of all the feasible solutions of the locus described in~\eqref{eq:locus}. The optimal position of the UAV $(x^*,h^*)$ is given by:

\textbf{Case \hspace{-.1mm}1)} If the Locus has no solution, the optimal position can be derived by solely examining the boundary positions:
  \begin{equation}
     (x^*,h^*)=\displaystyle \hspace{-15mm}\argmax_{(x,h)\in\{(0,h_{min}),(0,h_{max}),(D,h_{min}),(D,h_{max})\}}\hspace{-20mm} \textrm{SIR}_S(x,h).
    \end{equation}
    
\textbf{Case \hspace{-.1mm}2)} Upon having at least a feasible solution for the locus, if $\textrm{SIR}_2(D,h_{min})\leq \max \{ \textrm{SIR}_1(D,h_{max}),\textrm{SIR}_1(D,h_{min}) \}$ described by~\eqref{eq:case1}, $(x^*,h^*)=(D,h_{min})$. Also, if $\textrm{SIR}_1(0,h_{min})\leq \textrm{SIR}_2(0,h_{min})$ described by~\eqref{eq-b}, $(x^*,h^*)=(0,h_{min})$. Otherwise, let $(\tilde{x},\tilde{h}) = \argmax_{(x,h)\in \Lambda} \textrm{SIR}_S(x,h)$, then the optimal position of the UAV is given as follows:
\begin{itemize}
      \item If $\psi^x\geq D$: $x^*=\tilde{x}$ and $h^*$ can be derived using Theorem~\ref{th:givenX} considering $\hat{x}=\tilde{x}$.
    \item  If $\psi^x< D$ and $\tilde{x} \geq \psi^x$: $x^*=D$ and $h^*$ can be derived using Theorem~\ref{th:givenX} considering $\hat{x}=D$.
       \item If $\psi^x<  D$ and $\psi^h\leq \tilde{x} <\psi^x$ and $\textrm{SIR}_1(\tilde{x},\tilde{h})\geq\textrm{SIR}_1(D,h_{max})$: $(x^*,h^*)=(\tilde{x},\tilde{h})$.
     \item If $\psi^x<  D$ and $\psi^h\leq \tilde{x} <\psi^x$ and $\textrm{SIR}_1(\tilde{x},\tilde{h})<\textrm{SIR}_1(D,h_{max})$: $x^*=D$ and $h^*$ can be derived using Theorem~\ref{th:givenX} considering $\hat{x}=D$.
        \item If $\psi^x< D$ and $\tilde{x} <\psi^h$ and  $\textrm{SIR}_1(\tilde{x},h_{min})\geq \textrm{SIR}_1(D,h_{max})$: $(x^*,h^*)=(\tilde{x},h_{min})$.
         \item If $\psi^x< D$ and $\tilde{x} <\psi^h$ and  $\textrm{SIR}_1(\tilde{x},h_{min})< \textrm{SIR}_1(D,h_{max})$: $x^*=D$ and $h^*$ can be derived using Theorem~\ref{th:givenX} considering $\hat{x}=D$.
      \end{itemize}
\end{theorem}
\begin{proof}
The proof is an immediate result of studying the behaviors of the SIR expressions given in Lemma~\ref{lemma:lem1}.
\end{proof}
\subsection{Special Case}
Our derived expressions can be simplified to provide insights for various special situations. For example, suppose that the MSI is located on the segment between the Tx and the Rx ($Y_{_{\textrm{MSI}}}=0$), $\eta_{_{\textrm{NLoS}}}=\mu_{_{\textrm{NLoS}}}$,  and $p_t=p_u$. Considering~\eqref{eq:locus}, we get:
\begin{equation}
     B^2(x)-4A(x)C(x)=p_t^2
    4x^2(D-X_{_{\textrm{MSI}}})^2
  \geq 0.
\end{equation}
Normalizing the $p_t$ to $1$, the $\Lambda$ defined in~\eqref{eq:locus} is given by:
\begin{equation}
\begin{aligned}
    &\Lambda^+(x)= -x^2+2xD-DX_{_{\textrm{MSI}}},\\
    &\Lambda^-(x)=-x^2+2x X_{_{\textrm{MSI}}}-DX_{_{\textrm{MSI}}}.
    \end{aligned}
\end{equation}
The existence of a solution for~\eqref{eq:locus} requires $\Lambda^+(x)\geq 0$ or $\Lambda^-(x)\geq 0$, which is equivalent to:
\begin{align}
  &\hspace{-2.5mm} D-\sqrt{D^2-D X_{_{\textrm{MSI}}}} \leq x \leq D+ \sqrt{D^2-D X_{_{\textrm{MSI}}}}, \nonumber \\
  &\hspace{-2.5mm} X_{_{\textrm{MSI}}}\hspace{-.5mm}-\hspace{-.5mm}\sqrt{X_{_{\textrm{MSI}}}^2-D X_{_{\textrm{MSI}}}} \leq x \leq X_{_{\textrm{MSI}}}+\hspace{-.5mm} \sqrt{\hspace{-.5mm}X_{_{\textrm{MSI}}}^2-\hspace{-.5mm}D X_{_{\textrm{MSI}}}}.
\end{align}
It can be seen that the position of the MSI has a significant impact on these intervals and subsequently the placement of the UAV, especially if $X_{_{\textrm{MSI}}}\uparrow D$, it imposes $x\uparrow D$, and subsequently $h=\Lambda^+(x)=\Lambda^-(x)\downarrow 0$, which implies no feasible/practical solution for~\eqref{eq:locus}.\footnote{The notations $\uparrow$ and $\downarrow$ are used to denote approaching the limiting value from the left and the right, respectively.} Consequently, the optimal position is identified based on Case 1 of Theorem~\ref{th:var}.

Also, assuming that the source of interference is placed far away or it has a negligible transmission power, the interference will not play a key role in the design anymore. In this case, the SIR expressions in~\eqref{eq:SIRexp} will be replaced with signal-to-noise-ratio (SNR) expressions, which are much easier to handle compared to SIR expressions since they are monotone functions w.r.t both $h$ and $x$. In this case, a similar approach can be followed to obtain the optimal position of the UAV, which will result in simplified versions of Theorem~\ref{th:givenh}, \ref{th:givenX}, \ref{th:var}. The same philosophy holds for the following discussion on position planning for multiple UAVs.

\section{Position Planning for Multiple UAVs  Considering an MSI}\label{sec:multiple}
\noindent  We investigate the placement planning upon utilizing
multiple UAVs from two different points of view. First, we consider a cost effective design, in which the network designer aims to identify the minimum required number of utilized UAVs and determine their positions so as to satisfy a predetermined SIR of the system. Second, we assume that the network designer is provided with a set of UAVs, and endeavors to configure their positions so as to maximize the SIR of the system. {\color{black} We assume that the UAVs utilize the same frequency but different time slots to avoid mutual interference among the UAVs.}
\begin{figure}[t]
\includegraphics[width=8.8cm,height=2.1cm]{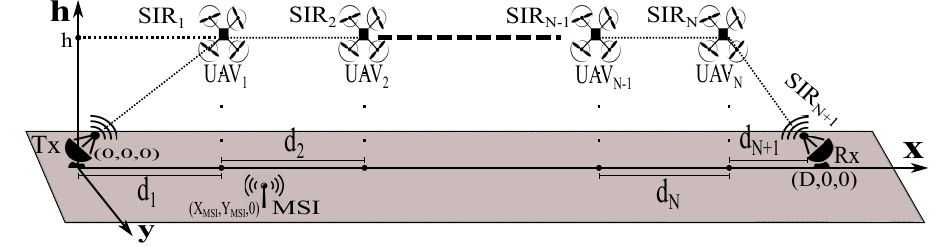}
		\caption{Multiple UAVs acting as relays between a pair of Tx and Rx coexisting with an MSI (multi-hop single link).}
		 \label{fig:multiple}
		 \end{figure}
 \vspace{-0mm}
\subsection{Network Design to Achieve a Desired SIR}\label{subsec:1}
Let $\gamma$ denote the desired SIR of the system and assume that $N$ is the minimum number of UAVs needed to satisfy the SIR constraint, which will be derived later. We index the Tx node by $0$, the UAVs between the Tx and the Rx from $1$ to $N$, and the Rx node by $N+1$. We denote the horizontal distance between two consecutive nodes $i-1$ and $i$ by $d_i$, $ 1\leq i\leq N+1$, and consider $\mathbf{d}=[d_1,\cdots,d_{N+1}]$. {\color{black}To have tractable derivations, we assume that all the UAVs have the same altitude $h$ (see Section~\ref{sec:multipleheight} for more details).} The model is depicted in Fig.~\ref{fig:multiple}. Let $\textrm{SIR}_k$ denote the SIR at the $k^{th}$ node, which can be obtained as:
\vspace{-0mm}
\begin{equation}\label{eq:SIRs}
\hspace{-1mm}
\begin{aligned}
  & \textrm{SIR}_1(\mathbf{d},h)= \frac{p_t \eta^{-1}_{_{\textrm{NLoS}}}\left(\sqrt{d_1^2+h^2}\right)^{-2}}{p_{_{\textrm{MSI}}}\eta^{-1}_{_{\textrm{NLoS}}}\left(\sqrt{(X_{_{\textrm{MSI}}}-d_1)^2+Y^2_{\textrm{MSI}}+h^2}\right)^{-2}},\\ 
  \vspace{-15mm} &\textrm{SIR}_2(\mathbf{d},h)=\frac{p_u \mu^{-1}_{_{\textrm{LoS}}}\left(\sqrt{d_2^2}\right)^{-2}}{\hspace{-3mm}p_{_{\textrm{MSI}}}\eta^{-1}_{_{\textrm{NLoS}}}\left(\sqrt{(X_{_{\textrm{MSI}}}\hspace{-1mm}-d_2-d_1)^2+\hspace{-1mm}Y^2_{\textrm{MSI}}\hspace{-1mm}+h^2}\right)^{-2}},\\\vspace{-1mm}
    \vdots\\\vspace{-1mm}
  & \textrm{SIR}_N(\mathbf{d},h)=\hspace{-0.5mm}\frac{p_u \mu^{-1}_{_{\textrm{LoS}}} \left(\sqrt{d_N^2}\right)^{-2}}{\hspace{-1.5mm}p_{_{\textrm{MSI}}}\eta^{-1}_{_{\textrm{NLoS}}}\hspace{-0.5mm}\left(\hspace{-0.5mm}\sqrt{\hspace{-0.5mm}(X_{_{\textrm{MSI}}}\hspace{-1.0mm}-\hspace{-0.5mm}\sum_{i=1}^{N} d_i)^2\hspace{-0.5mm}+\hspace{-0.5mm}Y^2_{\textrm{MSI}}\hspace{-0.5mm}+\hspace{-0.5mm}h^2\hspace{-0.5mm}}\right)^{\hspace{-0.5mm}-2}},\\ \vspace{-1mm}
  & \textrm{SIR}_{N+1}(\mathbf{d},h)=\frac{p_u\eta^{-1}_{_{\textrm{NLoS}}} \left(\sqrt{d_{N+1}^2+h^2}\right)^{-2}}{p_{_{\textrm{MSI}}}\mu^{-1}_{_{\textrm{NLoS}}}\left(\sqrt{(X_{_{\textrm{MSI}}}-D)^2+Y^2_{\textrm{MSI}}}\right)^{-2}}.
    \end{aligned}
 \end{equation}  
Similar to the single UAV scenario, $\textrm{SIR}_S $ is given by:
 \begin{equation}\label{eq:SIR_s2}
    \textrm{SIR}_S(\mathbf{d},h)=\min \big\{\textrm{SIR}_1(\mathbf{d},h),\cdots,\textrm{SIR}_{N+1}(\mathbf{d},h)\big\},~~\;\forall \mathbf{d},h.
\end{equation}

\subsubsection{The SIR expressions and the feasibility constraints}
From~\eqref{eq:SIRs}, it can be observed that achieving any desired $\textrm{SIR}_S$ ($\gamma$) may not be feasible. To derive the feasibility conditions for the $\gamma$, we need to analyze the links between the Tx and $\textrm{UAV}_1$, among the adjacent UAVs, and from $\textrm{UAV}_{N}$ to the Rx.

Analysis of the links between the Tx and $\textrm{UAV}_1$ ($\textrm{SIR}_1$) and between $\textrm{UAV}_{N}$ and the Rx ($\textrm{SIR}_{N+1}$) is similar to the discussion provided in Section~\ref{sec:singleUAV} (see Lemma~\ref{lemma:lem1}). Hence, we skip them and consider the SIR at $\textrm{UAV}_i$, $ 2\leq i\leq N$. For this UAV, the stationary point $\Phi_i^d$ of the SIR expression is given by:
\begin{equation}\label{eq:28}
    \Phi_i^d=\frac{h^2+Y_{_{\textrm{MSI}}}^2+(X_{_{\textrm{MSI}}}-\sum_{j=1}^{i-1} d_j)^2}{X_{_{\textrm{MSI}}}-\sum_{j=1}^{i-1} d_j},
\end{equation}
using which it can be validated that:
\begin{equation}\label{eq:dmiiidds}
\resizebox{0.89\hsize}{!}{$
\hspace{-15.5mm}
\displaystyle\max_{\hspace{13mm} x\in [0,D], h\in [h_{min},h_{max}]} \hspace{-11.5mm}\textrm{SIR}_i(x,h) \hspace{-.5mm}\leq\hspace{-.5mm} \frac{p_u\mu^{-1}_{_{\textrm{LoS}}}\hspace{-.5mm}\left(\max(X^2_{_{\textrm{MSI}}},(D-X_{_{\textrm{MSI}}})^2)+ Y_{_{\textrm{MSI}}}^2\hspace{-.5mm}+\hspace{-.5mm}h^2\right)}{p_{_{\textrm{MSI}}}\eta^{-1}_{_{\textrm{NLoS}}} d^2_{min} },\;$}
\end{equation}
\begin{equation}\label{eq:30}
\begin{cases}
    \frac{\partial \textrm{SIR}_i(x,h)}{\partial d_i}\hspace{-.5mm}\geq \hspace{-.5mm}0 \; \textrm{if} \; d_i\hspace{-.5mm}\geq\hspace{-.5mm}\Phi_i^d,\\
    \frac{\partial \textrm{SIR}_i(x,h)}{\partial d_i}\hspace{-.5mm}< \hspace{-.5mm}0 \; \textrm{O.W.},
    \end{cases}
\end{equation}
where $d_{min}$ is the minimum feasible distance between two UAVs considering the mechanical constrains. It can be verified that with tuning the locations of the middle UAVs any value for SIR$_i$ is achievable among the UAVs, when $\textrm{SIR}_i\leq \frac{p_u\mu^{-1}_{_{\textrm{LoS}}}\hspace{-.5mm}\left(Y_{_{\textrm{MSI}}}^2\hspace{-.5mm}+h^2\right)}{p_{_{\textrm{MSI}}}\eta^{-1}_{_{\textrm{NLoS}}} d^2_{min} }$. Combining these derivations with those in Section~\ref{sec:singleUAV}, we obtain the feasibility condition declared in~\eqref{eq:conditioforSIR}.
\begin{table*}[t]
\begin{minipage}{0.99\textwidth}
\begin{equation}\label{eq:conditioforSIR}
    \gamma \leq \min \Bigg\{\frac{p_t\left(X^2_{\textrm{MSI}}+Y_{\textrm{MSI}}^2+h_{min}^2 \right)}{p_{_{\textrm{MSI}}} h_{min}^2},\frac{p_u\mu^{-1}_{_{\textrm{LoS}}}\hspace{-.5mm}\left( Y_{_{\textrm{MSI}}}^2\hspace{-.5mm}+\hspace{-.5mm}h^2\right)}{p_{_{\textrm{MSI}}}\eta^{-1}_{_{\textrm{NLoS}}} d^2_{min} } , \frac{p_u\eta^{-1}_{_{\textrm{NLoS}}}\left(Y_{_{\textrm{MSI}}}^2+(D-X_{_{\textrm{MSI}}})^2\right)}{p_{_{\textrm{MSI}}}\mu^{-1}_{_{\textrm{NLoS}}}h_{min}^2}\Bigg\}
\end{equation}
\end{minipage}
\vspace{3mm}
\begin{minipage}{0.99\textwidth}
\vspace{3mm}
\begin{equation}\label{eq:d_1}
\begin{aligned}
  & d^+_1,d^-_1,=\frac{p_t X_{_{\textrm{MSI}}}\pm\sqrt{p^2_t X_{_{\textrm{MSI}}}^2-(p_t-\gamma p_{_{\textrm{MSI}}})\left(p_t\left(X_{_{\textrm{MSI}}}^2+Y_{_{\textrm{MSI}}}^2\right)+h^2(p_t-p_{_{\textrm{MSI}}}\gamma)\right)}}{p_t-\gamma p_{_{\textrm{MSI}}}}
    \end{aligned}
\end{equation}
\end{minipage}
\end{table*}
\subsubsection{Design Methodology}
To derive the minimum number of needed UAVs and their optimal positions so as to satisfy a desired $\textrm{SIR}_S$, we aim to maximize the distance between the UAVs while satisfying the desired SIR of the system. Our approach can be described by the following three main steps: (i) Considering $\textrm{SIR}_1$, for $\textrm{UAV}_1$, we obtain the maximum distance from the Tx (toward the Rx) ${d}^*_1$ that satisfies the SIR constraint and place the first UAV at the obtained location. (ii) Considering $\textrm{SIR}_{N+1}$, for $\textrm{UAV}_{N}$, we obtain the maximum distance from the Rx (toward the Tx)  ${d}^{max}$ that satisfies the desired SIR. (iii) Consider the segment with length $D-{d}^*_1-{d}^{max}$, we use the SIR expressions of the remaining UAVs to maximize the distance between the adjacent UAVs to cover the distance while satisfying the desired $\textrm{SIR}_S$. In the following, we explain these steps in more detail.

Considering $\textrm{SIR}_1$, we solve $\textrm{SIR}_1=\gamma$, the answer of which is given by~\eqref{eq:d_1}.
Then, using Lemma~\ref{lemma:lem1}, $d^*_1$ is given by:
\begin{equation}\label{eq:D1}
    d^*_1=\begin{cases}
              d_1^- \;\; \textrm{if}\;\; d_1^+> D,\\
            d_1^+ \;\; \textrm{if}\;\; d_1^+< \Psi^x \textrm{ and } d_1^+\leq D,\\
            D \;\; \textrm{O.W}.
    \end{cases}
\end{equation}
In the last case of \eqref{eq:D1}, the optimal number of UAVs is $1$, and the UAV should be placed at $x=D$. Assuming $d^*_1< D$, using Lemma~\ref{lemma:lem1}, {\color{black}the maximum distance between the last UAV and the Rx to satisfy the SIR constraint  $d^{max}$ can be obtained as:}
\begin{equation}\label{eq:DNplus1}
  d^{max}=\sqrt{\frac{p_u\eta^{-1}_{_{\textrm{NLoS}}}\left((X_{_{\textrm{MSI}}}-D)^2+Y_{_{\textrm{MSI}}}^2\right)}{\gamma p_{_{\textrm{MSI}}}\mu^{-1}_{_{\textrm{NLoS}}}}-h^2}.
\end{equation}
{\color{black}Note that using Lemma~\ref{lemma:lem1} it can be verified that if the distance between UAV$_{N}$ and Rx is less than $d^*_{N+1}$ the SIR constraint at the Rx is always met.}
Afterward, we solve $\textrm{SIR}_{k}=\gamma$ and use \eqref{eq:28}-\eqref{eq:30} to obtain $d^*_k$, $2\leq k \leq N$, given by:
\begin{equation}\label{eq:ds}
 d^*_k=\begin{cases}
      d_k^- \;\; \textrm{if}\;\; d_k^+> D-\sum_{j=1}^{k-1} d^*_j,\\
      d_k^+ \;\; \textrm{if}\;\; d_k^+< \Phi_k^d \textrm{ and } d_k^+\leq D-\sum_{j=1}^{k-1} d^*_j,\\
        D-d^{max} \;\; \textrm{O.W.},
  \end{cases}
\end{equation}
where $d^-_k,d^+_k$ are given in~\eqref{eq:d_k}.
\begin{table*}[t]
\begin{minipage}{0.99\textwidth}
\begin{equation}\label{eq:d_k}
\resizebox{0.95\hsize}{!}{$
          d^+_k,d^-_k=\frac{p_u \mu^{-1}_{_{\textrm{LoS}}}\left(X_{_{\textrm{MSI}}}-\sum_{i=1}^{k-1}d^*_i \right)\pm\sqrt{p^2_u \mu^{-2}_{\textrm{LoS}}\left(X_{_{\textrm{MSI}}}-\sum_{i=1}^{k-1} d^*_i \right)^2-p_u \mu^{-1}_{_{\textrm{LoS}}}\left(p_u \mu^{-1}_{_{\textrm{LoS}}}-\gamma p_{_{\textrm{MSI}}}\eta^{-1}_{_{\textrm{NLoS}}}\right)\left(h^2+Y_{_{\textrm{MSI}}}^2+\left(X_{_{\textrm{MSI}}}-\sum_{i=1}^{k-1} d^*_i \right)^2\right)}}{p_u \mu^{-1}_{_{\textrm{LoS}}}-\gamma p_{_{\textrm{MSI}}}\eta^{-1}_{_{\textrm{NLoS}}}}$}
\end{equation}
\begin{equation}\label{eq:sourecDefiner}
    x_H,y_H,p_H = \argmin_{x_H\in [0,D], y_H \in R^+, p_H\in R^+ } \int_{h_{min}}^{h_{max}} \int_{0}^{D} \abs[\Big]{\frac{p_H}{(x-x_H)^2+y^2_H+h^2} -\sum_{i=1}^{|I|}\frac{p_i}{(x-x_i)^2+y^2_i+h^2}} dx dh.
\end{equation}
\hrulefill
\end{minipage}
\end{table*}
Finally, the minimum number of required UAVs $N$ is given by:
\begin{equation}\label{eq:Opt}
    N=\argmin_{n\in\mathbb{N}} \sum _{k=2}^{n} d^*_k\geq D-d^*_1-d^{max}.
\end{equation}
According to~\eqref{eq:ds} and~\eqref{eq:d_k}, calculation of each $d^*_k$ only requires the knowledge of $d^*_{k'}$, $\forall k' < k$. Hence, the solution of \eqref{eq:Opt} can be easily obtained by initially assuming $n=2$ and increasing the value of $n$ by $1$ until the constraint in the right hand side of the equation is met. Given the distances $d^*_1,\cdots,d^*_{N}$ from~\eqref{eq:Opt}, we have $d^*_{N+1}=D-\sum_{k=1}^{N}d^*_k$, for which $d^*_{N+1}\leq d^{max}$.
\subsection{Distributed Position Planning for a Given Number of UAVs}\label{subDist}
In this case, there exist multiple UAVs dedicated as relays to the network, which are expected to be positioned to maximize the SIR of the system. To this end, 
an algorithm can be immediately proposed based on our results in the previous subsection, which considers the number of UAVs as given and slowly increases the SIR ($\gamma$) starting from $\gamma=0$ to find the maximum value of $\gamma$ for which $N$  in~\eqref{eq:Opt} becomes equal to the number of given UAVs. Afterward, the positions of the UAVs can be obtained as discussed before. Nevertheless, this is a centralized approach. In the following, we propose a distributed algorithm for the same purpose, where the UAVs locally compute their positions based on the knowledge of the positions of their adjacent neighbors, which can be obtained through simple message passing.
The following fact is an immediate consequence of examining \eqref{eq:d_k}: with a known $d_1$ and a (hypothetically) given value for the $\textrm{SIR}_S$ ($\gamma$), starting with $\textrm{UAV}_1$, the distance between the subsequent UAVs can be obtained up to $\textrm{UAV}_{N}$ in a \textit{forward propagation}, by which each UAV transmits its position to the adjacent UAV located toward the Rx (see~\eqref{eq:ds}, \eqref{eq:d_k}). Note that in the mentioned propagation, no message is exchanged between the last UAV and the Rx, and thus the SIR at the Rx might be less than $\gamma$. {\color{black}For this purpose the last UAV uses~\eqref{eq:DNplus1} to derive the maximum distance from the Rx for which the SIR constraint is met; and thus the last UAV can immediately verify the satisfaction of the SIR constraint given its current location.} We propose a distributed algorithm for position planning of multiple UAVs, the pseudo code of which is given in Algorithm~\ref{alg:fulldist}. In this algorithm, we first locate the first UAV above the Tx and the last UAV above the Rx and derive the initial desired $\textrm{SIR}_S$ ($\gamma^{(0)}$); subsequently, we set the position of the first UAV to have $\gamma^{(0)}$ as the SIR of the first link (lines~\ref{alg1line1}-\ref{alg1line3}). Afterward, using forward propagation, the UAVs locally obtain their positions w.r.t the position of their adjacent UAV (lines~\ref{alg1line4}) so as to satisfy the desired $\textrm{SIR}_S$. Then, the satisfaction of the SIR at the Rx is inspected via UAV$_N$ using the current distance from UAV$_N$ to the Rx (line~\ref{alg1line5}).  If this SIR satisfies the desired SIR of the system at the current iteration, the algorithm stops; otherwise, it starts over with a new desired value for $\textrm{SIR}_S$ and readjusts the locations of UAVs in the next iteration (lines~\ref{alg1line6}-\ref{alg1line7}).\footnote{{\color{black}We admit that in reality precise small adjustments of the locations of the UAVs may not be feasible due to physical hovering system limitations. However, this assumption is just needed to prove the optimality of the algorithm in theory. In reality, the error in the movements of the UAVs will result in a sub-optimal solution, which is unavoidable.}}
 \begin{algorithm}[h]
 	\caption{{Distributed position planning for multiple UAVs in the presence of a major source of interference}}\label{alg:fulldist}
 	\SetKwFunction{Union}{Union}\SetKwFunction{FindCompress}{FindCompress}
 	\SetKwInOut{Input}{input}\SetKwInOut{Output}{output}
 	 	{\small
 	\Input{Step size $\epsilon$.}
  $i=0$,
  $d^{(i)}_{1}=0,d^{(i)}_{N+1}=0$.\label{alg1line1}\\
  $\gamma^{(i)}=\min \big\{\textrm{SIR}_1(d^{(i)}_{1},h),\textrm{SIR}_{N+1}(d^{(i)}_{N+1},h)\big\}$.\label{alg1line2}\\	
  Derive $d^{(i)}_{1}$ for the target SIR $\gamma^{(i)}$ using~\eqref{eq:D1}.\label{alg1line3}\\
  Given $d^{(i)}_{1}$, obtain $d^{(i)}_{2},\cdots,d^{(i)}_{N}$ using forward propagation based on~\eqref{eq:ds} using $\gamma^{(i)}$.\label{alg1line4}\\
  Obtain the maximum distance $d^{max}$ from the Rx to satisfy the SIR constraint $\gamma^{(i)}$ using~\eqref{eq:DNplus1} at UAV$_N$.
  \label{alg1line5}\\
  \uIf{$D-\sum_{k=1}^{N} d^{(i)}_{k} > d^{max}$}{\label{alg1line6}
  Derive the next target SIR: $\gamma^{(i+1)}=\gamma^{(i)}-\epsilon$.\\
  $i=i+1$ and go to line~\ref{alg1line3}.
  }\Else{
  Fix the UAVs at their current positions.\\
  }\label{alg1line7}
  }
  \end{algorithm}
\subsubsection{Computational Complexity and Convergence Analysis}
At each iteration of our proposed distributed algorithm, in forward propagation mode, each UAV obtains its next location using a simple message passing with its adjacent UAV, through which the location of the adjacent UAV is exchanged, and calculation of a closed form expression (using~\eqref{eq:ds}) is performed. Hence, at each iteration, the computational complexity of the tasks performed at each UAV is $O(1)$. Also, it is obvious that, given a step size $\epsilon\leq \gamma^{(0)}$, the algorithm performs at most $\floor{\gamma^{(0)}/\epsilon}$ iterations. Thus, the worst computational complexity of our algorithm is $O(\gamma^{(0)}/\epsilon)$ at each UAV.
\begin{proposition}\label{prop:1}
 For any given number of UAVs and a sufficiently small size of the step size $\epsilon$, assuming the same altitudes for the UAVs, Algorithm~\ref{alg:fulldist} always converges to the maximum achievable value of $\textrm{SIR}_S$.
\end{proposition}
\begin{proof}
Assume that the number of UAVs is $N$ and the maximum achievable $\textrm{SIR}_S$ is $\gamma_{max}$. This implies the existence of a configuration of the UAVs in the sky corresponding to $\textrm{SIR}_S=\gamma_{max}$. Note that the proposed algorithm starts with the maximum attainable value for SIR$_S$, since the maximum value of SIR$_{N+1}$ is used to determine $\gamma^{(0)}$ and SIR$_S$ is the minimum SIR of all the links. Also, the algorithm makes small changes according to $\epsilon$ in the value of target SIR at each iteration. Thus, for $\epsilon\rightarrow 0$ there always exists an iteration $j$,  in which $\gamma^{(j)}=\gamma_{max}$. Note that achievability of $\gamma_{max}$ can be proved by contradiction. In other words, if the proposed maximum separation scheme does not find a configuration for the target SIR of $\gamma_{max}$, no other scheme would be able to do so. This is a direct result of forcing the maximum distance between the UAVs and the behavior of SIR expressions explained earlier. The iteration $j$ is indeed the last iteration of the algorithm since the constraint in line 9 of the algorithm will be met. 
\end{proof}
 {\color{black}
\subsection{Further adjustment of the altitudes and the horizontal distances of  UAVs}\label{sec:multipleheight}
The assumption made on the altitudes of the UAVs in the previous two subsections, i.e., the same altitude for all the UAVs, is common in current literature~\cite{chen2018multiple,sameAlt1,zhang2018trajectory,sameAlt2,sameAlt3}. Nevertheless, it can be predicted that adjustment of the altitude of each UAV can further increase the SIR of the system. Note that due to the non-convexity of the SIR expressions, proposing an analytical solution that jointly optimizes the altitudes and the horizontal positions of the UAVs in the multi-hop setting is hard to achieve. With this consideration, we propose a heuristic algorithm in Algorithm~\ref{alg:obtainHinDist}, which takes the horizontal distances of the UAVs obtained in the previous two subsections as the input and further adjusts the horizontal positions and altitudes of the UAVs.  
\begin{algorithm}[h]
  {\color{black}
 	\caption{Heuristic joint altitude and horizontal distance adjustments  for multiple UAVs in the presence of a major source of interference}\label{alg:obtainHinDist}
 	\SetKwFunction{Union}{Union}\SetKwFunction{FindCompress}{FindCompress}
 	\SetKwInOut{Input}{input}\SetKwInOut{Output}{output}
 	 	{\footnotesize
 	\Input{Vertical exploration $\epsilon_h$, initial horizontal distances of the UAVs $\mathbf{d}^{(0)}=[d^{(0)}_1,d^{(0)}_2,\cdots,d^{(0)}_{N+1}]$, where $d^{(0)}_i$ is the initial horizontal distance of UAV$_i$, initial altitudes of the UAVs $\mathbf{h}^{(0)}=[{h}^{(0)}_{(1)},\cdots,{h}^{(0)}_{(N)}]$, where ${h}^{(0)}_{(i)}$ is the initial altitude of UAV$_i$, number of iterations $Iter_{max}$, safe-guard distance between UAVs $d_{min}$.}
  $k=0$.\\
  \For{$k=1$ to $Iter_{max}$}{
  \For{$i=1$ to $N$}{
  For UAV$_i$, consider the two local links passing through the UAV and obtain their SIRs for the current network configuration, i.e., $\textrm{SIR}^{(k)}_{i}(\mathbf{d}^{(k)},\mathbf{h}^{(k)}),\textrm{SIR}^{(k)}_{i+1}(\mathbf{d}^{(k)},\mathbf{h}^{(k)})$ using the SIR expressions given below: 
\begin{equation}
\hspace{-5.5mm}
\begin{aligned}
  & \textrm{SIR}_1(\mathbf{d},\mathbf{h})= \frac{p_t \eta^{-1}_{_{\textrm{NLoS}}}\left(\sqrt{d_1^2+h_{1}^2}\right)^{-2}}{p_{_{\textrm{MSI}}}\eta^{-1}_{_{\textrm{NLoS}}}\left(\sqrt{(X_{_{\textrm{MSI}}}-d_1)^2+Y^2_{\textrm{MSI}}+h_{1}^2}\right)^{-2}},\\ 
  &    \textrm{SIR}_i(\mathbf{d},\mathbf{h})=\hspace{-0.5mm}\frac{p_u \mu^{-1}_{_{\textrm{LoS}}} \left(\sqrt{d_i^2+(h_{i}-h_{i-1})^2}\right)^{-2}}{\hspace{-1.5mm}p_{_{\textrm{MSI}}}\eta^{-1}_{_{\textrm{NLoS}}}\hspace{-0.5mm}\left(\hspace{-0.5mm}\sqrt{\hspace{-0.5mm}(X_{_{\textrm{MSI}}}\hspace{-1.0mm}-\hspace{-0.5mm}\sum_{i=1}^{N} d_i)^2\hspace{-0.5mm}+\hspace{-0.5mm}Y^2_{\textrm{MSI}}\hspace{-0.5mm}+\hspace{-0.5mm}h^2_i\hspace{-0.5mm}}\right)^{\hspace{-0.5mm}-2}},~2\leq i\leq N,\\ \vspace{-1mm}
  & \textrm{SIR}_{N+1}(\mathbf{d},\mathbf{h})=\frac{p_u\eta^{-1}_{_{\textrm{NLoS}}} \left(\sqrt{d_{N+1}^2+h_{N}^2}\right)^{-2}}{p_{_{\textrm{MSI}}}\mu^{-1}_{_{\textrm{NLoS}}}\left(\sqrt{(X_{_{\textrm{MSI}}}-D)^2+Y^2_{\textrm{MSI}}}\right)^{-2}}.
    \end{aligned}
 \end{equation}  
\label{line4Alg2}\\
  Define the local SIR at UAV\_{i} as: $\textrm{SIR}^{Local}_{i}(\mathbf{d}^{(k)},\mathbf{h}^{(k)})\triangleq \min\{\textrm{SIR}^{(k)}_{i}(\mathbf{d}^{(k)},\mathbf{h}^{(k)}),\textrm{SIR}^{(k)}_{i+1}(\mathbf{d}^{(k)},\mathbf{h}^{(k)})\}$. \label{line5Alg2}\\
  Confine the horizontal location of UAV$_{i}$ between the previous node and the next node of the network and obtain the best horizontal location $d^{*}_i$ and altitude $h^{*}_i$ of UAV$_{i}$ as: $(d^{*}_i,h^{*}_i)=\hspace{-1cm}\displaystyle\argmax_{0\leq d \leq d^{(k)}_i+d^{(k)}_{i+1},\;\; h^{(k)}_i-\epsilon_h \leq h \leq h^{(k)}_i+\epsilon_h} {\textrm{SIR}^{Local}_{i}(\mathbf{\tilde{d}}^{(k)},\mathbf{\tilde{h}}^{(k)})}$, where $\mathbf{\tilde{d}}^{(k)}$ and $\mathbf{\tilde{h}}^{(k)}$ are vectors $\mathbf{{d}}^{(k)}$ and $\mathbf{{h}}^{(k)}$ with their $i$-th element replaced by the variables under optimization, i.e., $d$ and $h$, respectively; also $i+1$-th element of $\mathbf{\tilde{d}}^{(k)}$ is replaced with ${d}^{(k)}_{i+1}+{d}^{(k)}_{i}-d$. Note that $\sqrt{d^2+(h-h^{(k)}_{i-1})^2}\geq d_{min}$ and $\sqrt{(d^{(k)}_{i+1}+{d}^{(k)}_{i}-d)^2+(h-h^{(k)}_{i+1})^2}\geq d_{min}$ are inherently assumed to avoid violating the required safe-guard distance between the UAV and the previous and the next node.\label{line6Alg2}\\
  Consider vectors $\mathbf{\bar{d}}^{(k)}$ and $\mathbf{\bar{h}}^{(k)}$ as vectors $\mathbf{{d}}^{(k)}$ and $\mathbf{{h}}^{(k)}$ with their $i$-th element replaced by the solution of the optimization, i.e., $d^*_i$ and $h^*_i$, respectively; where $i+1$-th element of $\mathbf{\bar{d}}^{(k)}$ is ${d}^{(k)}_{i+1}+{d}^{(k)}_{i}-d^*_i$.\\
  \uIf{$\textrm{SIR}^{Local}_{i}(\mathbf{\bar{d}}^{(k)},\mathbf{\bar{h}}^{(k)})>\textrm{SIR}^{Local}_{i}(\mathbf{{d}}^{(k)},\mathbf{{h}}^{(k)})$}{
  $d^{(k)}_i=d^{*}_i,\;\; h^{(k)}_i=h^{*}_i$\label{line7Alg2}
  }
  }
    $\mathbf{h}^{(k+1)}=\mathbf{h}^{(k)}$, $\mathbf{d}^{(k+1)}=\mathbf{d}^{(k)}$.
  }
  }
  }
  \end{algorithm}

The basic idea behind our algorithm is to increase the minimum SIR of all the adjacent pairs of links, which connect the Tx to the Rx, at each iteration. Considering~\eqref{eq:SIR_s2}, this results in increasing the SIR of the system at each iteration. At each iteration of  Algorithm~\ref{alg:obtainHinDist}, each UAV aims to increase the minimum SIR of the two local links, i.e., the link connecting the previous node of the network to the UAV and the link connecting the UAV to the next node, that pass through it (lines~\ref{line4Alg2},\ref{line5Alg2}). To this end, at each iteration, starting from the first UAV, each UAV explores the horizontal distances between itself and the two adjacent nodes and the range of altitudes specified by a the vertical exploration parameter $\epsilon_h$ while the location of all the other UAVs are fixed  (line~\ref{line6Alg2}). In other words, each UAV, $u_i$, with height $h_i$ explores a rectangle area in the space, the width of which is the distance between its two adjacent nodes and the height of which is $2\epsilon_h$ ($h$ is in the range $[h_i-\epsilon_h,h_i+\epsilon_h]$). Then, it changes its location to the best found position corresponding to the maximum increase in the minimum SIR of its local links (line~\ref{line7Alg2}). Afterward, the next UAV aims to achieve the same goal considering the modified position of the previous UAV and the process continues until the last UAV adjusts its location, which completes one iteration. The following three facts are immediate: i) the algorithm is convergent. ii) the SIR of the system achieved at each iteration of the algorithm is non-decreasing w.r.t the iteration number. In other words, the algorithm always increases the SIR of the system by a value that is greater than or equal to zero from one iteration to the next. iii) Given the fact that through the position adjustment process, each UAV requires the locations of the two adjacent nodes and only conducts internal calculations, the algorithm can be implemented in a distributed fashion. 
  \section{Extension to Multiple Sources of Interference}
Pursuing a similar approach to the previous two sections, i.e., obtaining the positions of the UAVs w.r.t the position of the MSI, upon existence of multiple sources of interference (SI) is highly challenging. Nevertheless, in the following, we demonstrate that the results derived in the previous sections for a single MSI can be easily extended to this scenario to provide approximate solutions. Let $I=\{\textrm{SI}_1,\textrm{SI}_2,\cdots,\textrm{SI}_{|I|}\}$ denote the set of SIs. Let $(x_i,y_i,0)$ denote the position of $\textrm{SI}_i$, $\forall i$. Assuming that the interference is superimposed, at any given point in the sky, $h>0$, the total interference power at point $(x,y,h)$, $\kappa(x,y,h)$, is given by:
\begin{equation}\label{eq:effectMultiple}
    \kappa (x,y,h)=\sum_{i=1}^{|I|}\frac{p_i/\eta_{_{\textrm{NLoS}}}}{(x-x_i)^2+y^2_i+h^2}.
\end{equation}
We model the effect of all the SIs as a single hypothetical MSI. For this purpose, we assume that the hypothetical MSI is placed at $(x_H,y_H,0)$ with power $p_H$. The hypothetical MSI should exhibit a similar interference effect as the SIs in the sky. Hence, we formulate the problem of obtaining $x_H,y_H,p_H$ as minimizing the approximation error described in~\eqref{eq:sourecDefiner} or the discretized version of it, i.e., replacing the integrals in~\eqref{eq:sourecDefiner} with summations, both of which can be solved numerically. Another approach is to approximate~\eqref{eq:effectMultiple} and the interference power expression of the hypothetical MSI using their Taylor expansions and obtain $x_H,y_H,p_H$ accordingly. Note that the approximation error depends on the positions and transmitting powers of the SIs. In general, the approximation error is lower when the SIs are closer to each other and have more homogeneous transmitting powers. {\color{black} Nevertheless, when the mentioned criterion  is not satisfied, the positions of the UAVs can be obtained by measuring the interference power in the environment, which is further discussed in the next section.
}
\section{Stochastic Interference}
In some scenarios, the position of the MSI or the SIs are not known. Also, the number of interference sources, their positions, and their {\color{black} transmission powers} may change over time, and thus not be fixed. In these scenarios, the interference can be considered as a random variable at each point in the sky. Deriving the distribution of this random variable at any point requires measuring the interference power at that point for a long time. In the following discussion, to have tractable derivations, we fix the altitude for both the single UAV and multiple UAV scenarios. {\color{black} Note that for both dual-hop and multi-hop settings a similar approach proposed in Algorithm~\ref{alg:obtainHinDist} can be used to further adjust the altitude(s) and the horizontal position(s) of the UAV(s) upon having stochastic interference.} In the following, we study the position planning for a single UAV and multiple UAVs considering stochastic interference in order.

  }

\subsection{Position planning for a single UAV considering stochastic interference}
Assuming the altitude of the UAV to be $h$, let random variable $I_{a}$ denote the power of interference at horizontal position $x=a$ with the corresponding probability density function (pdf) $f_{I_{a}}(y)$ and the \textit{moment generating function} $M_{I_{a}}(y)= E (\exp(yI_{a}))$.\footnote{Note that due to fixed altitude, for a better readability, index $h$ is omitted from the interference related terms such as $I$ and later $\Upsilon$ since the interference is assumed to be measured in that altitude.} Throughout, we assume that the $I_{.}$-s at different horizontal positions are independent. In this case, the SIR expressions given in~\eqref{eq:SIRexp} will become random variables defined as follows:
\begin{equation}\label{SIRstoc}
\begin{aligned}
&\textrm{SIR}_1(x,h)\hspace{-1mm}=\hspace{-1mm}\frac{p_t\hspace{-.14mm}/\hspace{-.15mm}\left(\hspace{-.15mm}\eta_{_{\textrm{NLoS}}} (x^2+h^2)\hspace{-.15mm}\right)}{I_x},\\
    &  \textrm{SIR}_2\hspace{-.1mm}(\hspace{-.15mm}x,h\hspace{-.15mm})\hspace{-1mm}= \hspace{-.199mm}\frac{\hspace{-.19mm}p_u\hspace{-.15mm}/\hspace{-.15mm}\left(\hspace{-.15mm}\eta_{_{\textrm{NLoS}}} \hspace{-.15mm}\left((D-x)^2+h^2\right)\hspace{-.15mm}\right)}{I_D}.
      \end{aligned}
\end{equation}


As a reasonable extension of~\eqref{eq:SIR_s}, we opt to work with the expected value of the SIR expressions: 
\begin{equation}
\label{eq:SIR_s3}
  \hspace{-4mm}   \overline{\textrm{SIR}_S}(x,h)= \min \big\{E(\textrm{SIR}_1(x,h)),E(\textrm{SIR}_2(x,h))\big\}, \;\forall x,h. \hspace{-6mm} 
\end{equation}
Considering~\eqref{SIRstoc}, to derive the expected value of the SIR expressions, we first derive the expected value of the inverse of the interference random variable as:
\begin{equation}\label{eq:expectedCalc}
 \begin{aligned}
     &E(\frac{1}{I_{a}})=\int_{0}^{\infty} \frac{1}{x} f_{I_{a}}(x) dx\\&=\int_{0}^{\infty} \hspace{-2mm}\int_{0}^{\infty} \exp(-yx) f_{I_{a}}(x)dy  dx =\int_{0}^{\infty} M_{I_{a}}(-y)  dy.
     \end{aligned}
 \end{equation}
Define $\Upsilon_a \triangleq \int_{0}^{\infty} M_{I_a}(-y)  dy$, whose calculation is deferred to Section~\ref{sec:betaDistribution}. Then the expected value of the SIR expressions are given by:
 \begin{equation}\label{eqImp}
\begin{aligned}
&E(\textrm{SIR}_1(x,h))=\Upsilon_x p_t/\left(\eta_{_{\textrm{NLoS}}} (x^2+h^2)\right),\\
    &  E(\textrm{SIR}_2(x,h))=\Upsilon_D p_u/\left(\eta_{_{\textrm{NLoS}}}\left((D-x)^2+h^2\right)\right).
      \end{aligned}
\end{equation}
Similar to Lemma~\ref{lemma:lem1}, it can be verified that $E(\textrm{SIR}_2(x,h))$ is a monotone increasing function w.r.t $x$, $x\in(0,D)$. Nevertheless, $E(\textrm{SIR}_1(x,h))$ exhibits different behaviors for different moment generating functions. 
 Due to this fact, deriving the analytic optimal solution in this case is intractable. Instead, we propose the following iterative approach to solve the problem. 
 Assume that we need to obtain the position of the UAV to satisfy a given $\overline{\textrm{SIR}_S}(x,{h})$ denoted by $\gamma$. Using the $E(\textrm{SIR}_2(x,h))$  expression in~\eqref{eqImp}, the corresponding value of $x$ is the feasible answer, i.e., belonging to $[0,D]$, of the following equation:
 \begin{equation}\label{eq:fixedhStoc}
     x= D - \sqrt{\frac{\Upsilon_D p_u}{\eta_{_{\textrm{NLoS}}}\gamma}-{h}^2}.
 \end{equation}
 Using the facts that the expression under the square root should be positive and $x$ should belong to $[0,D]$, the feasible values of the $\gamma$ are bounded by:
 \begin{equation}\label{eq:fixedhStoccc}
  \frac{\Upsilon_D p_u}{\eta_{_{\textrm{NLoS}}}(D^2+{h}^2)}   \leq  \gamma \leq \frac{\Upsilon_D p_u}{\eta_{_{\textrm{NLoS}}}{h}^2}.
 \end{equation}
 Define $\gamma_{min}\triangleq \frac{\Upsilon_D p_u}{\eta_{_{\textrm{NLoS}}}(D^2+{h}^2)} $ and $ \gamma_{max} \triangleq \frac{\Upsilon_D p_u}{\eta_{_{\textrm{NLoS}}}{h}^2}$.  At each iteration, $i$, of our iterative algorithm described in Algorithm~\ref{alg:fulldistStochastic2}, the algorithm derives the horizontal position of the UAV, $x^{(i)}$, to satisfy a given $\overline{\textrm{SIR}_S}$, i.e., $\gamma^{(i)}$, by solely considering the $E(\textrm{SIR}_2(x^{(i)},{h}))$ using \eqref{eq:fixedhStoc}. Afterward, it checks the $\overline{\textrm{SIR}_S}$ constraint using $E(\textrm{SIR}_1(x^{(i)},{h}))$. If the chosen $\gamma^{(i)}$ was unfeasible, the algorithm decreases the value of the targeted $\overline{\textrm{SIR}_S}$ by a tunnable step size ($\epsilon$) for the next iteration. {\color{black}It can be verified that the algorithm converges in at most $\floor{(\gamma_{max}-\gamma_{min})/\epsilon}$ iterations.}
 
{\color{black} \begin{proposition}\label{prop:optAlog2}
 For sufficiently small values of $\epsilon$, Algorithm~\ref{alg:fulldistStochastic2} identifies the optimal position of the UAV.
\end{proposition}
\begin{proof}
The initial point of Algorithm~\ref{alg:fulldistStochastic2} is $\gamma_{max}$, i.e., the maximum value of $E(\textrm{SIR}_2(x^{(i)},{h}))$, which coincides with the maximum attainable value of $\overline{\textrm{SIR}_S}(x,h)= \min \big\{E(\textrm{SIR}_1(x,h)),E(\textrm{SIR}_2(x,h))\big\}$. For sufficiently small values of $\epsilon$, the algorithm gradually decreases the desired $\gamma$ and explores all the possible locations of the UAV, which corresponds to the interval of $\gamma$, where $\gamma_{min}\leq  \gamma \leq \gamma_{max}$. The proof is complete since $E(\textrm{SIR}_2(x,h))$ is a monotone increasing function w.r.t $x$ and for the values of $\gamma$-s outside of this interval $E(\textrm{SIR}_2(x,h))$ is not feasible.
\end{proof}}
 
 \begin{algorithm}[t]
 	\caption{Iterative approach to obtain the optimal horizontal position of the UAV under the presence of stochastic interference}\label{alg:fulldistStochastic2}
 	\SetKwFunction{Union}{Union}\SetKwFunction{FindCompress}{FindCompress}
 	\SetKwInOut{Input}{input}\SetKwInOut{Output}{output}
 	 	{\footnotesize
 	\Input{Step size $\epsilon$.}
  $i=0$,
  Choose $\gamma^ {(0)}=\gamma_{max}$.\\
  Derive the position of the UAV, i.e., $x^{(i)}$,  using~\eqref{eq:fixedhStoc} with $\gamma^ {(i)}$. \label{lineinAlg1Stoc2}\\
  \uIf{$E(\textrm{SIR}_1(x^{(i)},{h})) < \gamma^ {(i)}$}{
  $\gamma^ {(i+1)}=\gamma^ {(i)}-\epsilon$\\
  \uIf{$\gamma^ {(i+1)}\leq \gamma_{min}$}{
    Among the previously investigated horizontal positions, fix the UAV at position ${{x^{(i')}}^*}$, where ${{(i')}^*}=\displaystyle \argmax_{i'\in \{1,2,\cdots,i \}} \min \{\gamma^{(i')}, E(\textrm{SIR}_1(x^{(i')},{h}))\}$
  }\Else{
  $i=i+1$ and go to line~\ref{lineinAlg1Stoc2}.
  }
  }\Else{
  Fix the UAVs at its current position.~\label{lineinAlg1Stoc3}
  }
  }
  \end{algorithm}

\subsection{Distributed Position planning for Multiple UAVs considering stochastic interference}
Assuming multiple UAVs as described in Section~\ref{sec:multiple}, upon having a stochastic interference, the SIR expressions given by~\eqref{eq:SIRs} will be random variables described as follows:
\begin{equation}\label{eq:SIRsStoc}
\hspace{-1mm}
\begin{aligned}
  & \textrm{SIR}_1(\mathbf{d},h)= \frac{p_t \eta^{-1}_{_{\textrm{NLoS}}}\left(\sqrt{d_1^2+h^2}\right)^{-2}}{I_{d_1}},\\ 
  \vspace{-15mm} 
    \vdots\\\vspace{-1mm}
  & \textrm{SIR}_N(\mathbf{d},h)=\hspace{-0.5mm}\frac{p_u \mu^{-1}_{_{\textrm{LoS}}} \left(\sqrt{d_N^2}\right)^{-2}}{I_{\sum_{i=1}^{N} d_i}},\\ \vspace{-1mm}
  & \textrm{SIR}_{N+1}(\mathbf{d},h)=\frac{p_u\eta^{-1}_{_{\textrm{NLoS}}} \left(\sqrt{d_{N+1}^2+h^2}\right)^{-2}}{I_{D}}.
    \end{aligned}
 \end{equation}  
 In this case,
  \begin{equation}
   \hspace{-.1mm}\overline{\textrm{SIR}_S(\mathbf{d},h)}=
   \min \big\{E (\textrm{SIR}_1(\mathbf{d},h)),\cdots, E (\textrm{SIR}_{N+1}(\mathbf{d},h))\big\},~~\forall \mathbf{d},h.\hspace{-2mm}
\end{equation}
 Assume that the random variable $I_{\sum_{i=1}^{k} d_i}$ follows the distribution $f_{_{I_{\sum_{i=1}^{k} d_i}}}(y)$ with moment generating function $M_{_{I_{\sum_{i=1}^{k} d_i}}}(y)= E (\exp(yI_{\sum_{i=1}^{k} d_i}))$. Considering~\eqref{eq:expectedCalc}, the expected value of the SIR expressions can be obtained as:

\begin{equation}\label{eqImp2}
\begin{aligned}
&E(\textrm{SIR}_1(\mathbf{d},h))=\Upsilon_{d_1}p_t/\left(\eta_{_{\textrm{NLoS}}} \left(d_1^2+h^2\right)\right),\\
&\vdots\\
&E(\textrm{SIR}_{N}(\mathbf{d},h))=\Upsilon_{\sum_{i=1}^{N} d_i} p_u/\left(\mu_{_{\textrm{LoS}}} d_{N}^2 \right),\\
    &  E(\textrm{SIR}_{N+1}(\mathbf{d},h))=\Upsilon_D p_u/\left(\eta_{_{\textrm{NLoS}}}\left(d_{N+1}^2+h^2\right)\right),
      \end{aligned}
\end{equation} 
In the following, we investigate the two problems pursued in Sections \ref{subsec:1} and \ref{subDist} considering the stochastic interference.

\subsubsection*{\textbf{Problem 1: Determining the minimum required number of UAVs and their locations to achieve a desired $\overline{\textrm{SIR}_S}$}} Considering~\eqref{eqImp2}, upon having non-identical interference distributions along the $x$-axis, the SIR expressions can exhibit different behaviors in various horizontal positions. Hence,  in general, the exact analysis of this problem is intractable in this case, and thus it can only be solved by exhaustive search, which can be computationally prohibitive. Considering this fact, pursuing a similar approach to Section~\ref{subsec:1}, we propose a sub-optimal approach based on maximizing the distances between the UAVs to cover the span between the Tx and the Rx that guarantees achieving the desired $\overline{\textrm{SIR}_S}$ ($\gamma$), and obtains the required number of UAVs and their locations. Note that, in general, the obtained number of UAVs using our approach may not always meet the minimum number of UAVs needed to satisfy the desired $\overline{\textrm{SIR}_S}$ determined by the exhaustive search. However, we will show in Proposition~\ref{prop:minNum} that these two numbers collide when the interference is i.i.d. along the $x$-axis.

Let $N$ denote the minimum number of needed UAVs, and, correspondingly, $d^*_1,\cdots,d^*_{N}$ denote the distances between the UAVs, both of which will later be derived. Let $\mathcal{D}_1$ denote the set of solutions of  $E\left(\textrm{SIR}_1(\mathbf{d},h)\right)\geq \gamma$ given by:
 \begin{equation}\label{d1Stoc}
     {d}^2_1\gamma - p_t \eta^{-1}_{_{\textrm{NLoS}}} \Upsilon_{d_1}+h^2\gamma\leq 0, ~~\forall d_1\in \mathcal{D}_1,
 \end{equation}
 which can be numerically obtained. Note that $d^*_1\in\mathcal{D}_1$. Also, {\color{black} the maximum distance between the Rx and UAV$_{N}$ to satisfy the SIR constraint $d^{max}$} is the following closed form expression obtained from solving $E\left(\textrm{SIR}_{N+1}(\mathbf{d},h)\right)= \gamma$:
  \begin{equation}\label{dkClose}
d^{max} = \sqrt{ \frac{p_u\eta^{-1}_{_{\textrm{NLoS}}}}{\gamma} \Upsilon_D - h^2}.
 \end{equation}
If $\exists d_1\in\mathcal{D}:  d^{max}\geq D-d_1$, using a UAV is enough to achieve the desired SIR$_S$. Otherwise, the position of the middle UAVs, $d^*_{k}$, $2\leq k\leq N$, is the largest solution to the following equation:
 \begin{equation}\label{eq:num1}
{d^*}^2_{k} \gamma-p_u \mu^{-1}_{_{\textrm{LoS}}}\Upsilon_{\sum_{i=1}^{k} d^*_i} =0.
 \end{equation}
 It can be seen that obtaining the distances using the sequence $d^*_1\rightarrow d^*_2\rightarrow d^*_3 ...$ requires solving~\eqref{eq:num1} numerically since $d^*_{k}$ appears in the argument of the $\Upsilon$, which may significantly reduce the speed of computations. Nevertheless, except for $d^*_1$, all the distances between UAVs can be obtained in a closed form expression by pursuing the following backward approach. Since $\sum_{i=1}^{k} d^*_i=D-\sum_{i=k+1}^{N+1}d^*_i$, considering $d^*_{N+1}=d^{max}$ given by the closed form expression in~\eqref{dkClose}, with replacing $\sum_{i=1}^{k} d^*_i$ by $D-\sum_{i=k+1}^{N+1} d^*_i$ in~\eqref{eq:num1}, we can obtain the position of the middle UAVs in a backward order $d^*_{N+1}\rightarrow d^*_{N}\rightarrow d^*_{N-1} ...$ using the following expression:
 \begin{equation}\label{eq:num}
d^*_{k}=\sqrt{p_u \gamma^{-1}\mu^{-1}_{_{\textrm{LoS}}}\Upsilon_{D-\sum_{i=k+1}^{N+1} d^*_i}}\;,\;\; 2\leq k\leq N.
 \end{equation}
Finally, the minimum required number of UAVs $N$ using our approach is given by:
\begin{equation}\label{eq:Opt2}
    N=\min \big\{n\in\mathbb{N}:~\exists d_1\in\mathcal{D}, ~\sum _{k=2}^{n} d^*_k= D-d_1-d^{max} \big\},
\end{equation}
which can be solved similarly to~\eqref{eq:Opt}. Since for all $ d^*_{N+1}$, where $d^*_{N+1}\leq d^{max}$, the SIR constraint at the Rx is met, if the above equation has no solution, its right hand side can be replaced by $D-d_1-(d^{max}-\rho)$, where $d^{max}-\rho\geq 0$ determines the position of the last UAV, i.e., $d^*_{N+1}=d^{max}-\rho$. Then, starting form $\rho=0$, the problem can be solved assuming small increments in the value of $\rho$ until a solution is found.
\begin{proposition}\label{prop:minNum}
 The determined number of UAVs $N$ using our approach given in \eqref{eq:Opt2} coincides with the minimum required number of UAV to satisfy a desired $\overline{\textrm{SIR}_S}$, when the interference is i.i.d. through the $x$-axis.
\end{proposition}
\begin{proof}
For i.i.d. interference through $x$-axis, $\Upsilon_x$ is identical for $x\in[0,D]$. Thus, considering~\eqref{eqImp2}, $E(\textrm{SIR}_{i}(x,h))$ is a monotone decreasing function w.r.t $d_i$, $\forall i$. Hence, the distance between the UAVs obtained using our method coincides with the maximum distance between UAVs to satisfy the desired $\overline{\textrm{SIR}_S}$, given which the proposition result is immediate.
\end{proof}
\subsubsection*{\textbf{Problem 2: Obtaining a distributed algorithm to increase $\overline{\textrm{SIR}_S}$ for a given number of UAVs in the system}} Considering our methodology in the previous problem and the method described in Section~\ref{subDist}, we propose Algorithm~\ref{alg:fulldistStochastic} to increase the $\overline{\textrm{SIR}_S}$, which only uses message passing between adjacent UAVs to exchange their current location. As can be seen, the proposed algorithm is similar to the proposed algorithm in Subsection~\ref{subDist}. The only difference is that, instead of forward propagation, Algorithm~\ref{alg:fulldistStochastic} uses \textit{backward propagation}, by which each UAV transmits its position to the adjacent UAV located toward the Tx, which is used to exploit the closed form expression given by~\eqref{eq:num}. Consequently, to examine the satisfaction of the targeted $\overline{\textrm{SIR}_S}$, the average value of the SIR at the first UAV is measured. Similar to Algorithm~\ref{alg:fulldist}, Algorithm~\ref{alg:fulldistStochastic} converges in at most $\floor{\gamma^{(0)}/\epsilon}$ iterations, and its worst computational complexity is $O(\gamma^{(0)}/\epsilon)$ at each UAV. In the following, we obtain the performance guarantee of our proposed algorithm for both i.i.d. and non-i.i.d interference through the horizontal axis. 

\begin{proposition}
 For any given number of UAVs and a sufficiently small step size $\epsilon$, assuming the same altitudes for the UAVs, if $I_x$-s are i.i.d. for $x\in[0,D]$, then Algorithm~\ref{alg:fulldistStochastic} is guaranteed to converge to the maximum achievable value of $\overline{\textrm{SIR}_S}$.
\end{proposition}
\begin{proof}
Assume that given the number of UAVs $N$, $\gamma^{max}$ is the maximum achievable $\overline{\textrm{SIR}_S}$. The proof is similar to the proof of Proposition~\ref{prop:1} considering the following facts: i) for i.i.d. interference through $x$-axis, $\Upsilon_x$ will be identical for $x\in[0,D]$; ii) considering~\eqref{eqImp2}, in this case, $E(\textrm{SIR}_{i}(x,h))$ will be a monotone decreasing function w.r.t $d_i$, $\forall i$; iii) the algorithm aims to achieve the target SIR at each iteration via creating the maximum distance between the UAVs.
\end{proof}

\begin{proposition}
 For a given number of UAVs, a sufficiently small step size $\epsilon$, assuming the same altitudes for the UAVs, if $I_x$-s are non-i.i.d. for $x\in[0,D]$,  then Algorithm~\ref{alg:fulldistStochastic} is guaranteed to converge to the maximum value of sequence $\{ \gamma^{(i)}\}_{i=0}^{\floor{\gamma^{(0)}/\epsilon}}$, where $\gamma^{(i)}=\gamma^{(0)}-i\times\epsilon $, for the target SIR$_S$, for which covering the distance between the Tx and the Rx using backward propagation and maximum distance between the UAVs obtained in \eqref{eq:num} is feasible.
\end{proposition}
\begin{proof}
The proof is trivial since the algorithm makes small changes according to the step size $\epsilon$ in the target SIR$_S$. Also, the algorithm stops at iteration $i$ when the current target value of SIR$_S$ ($\gamma^{(i)}$) is feasible and  the next target value for the SIR$_S$ ($\gamma^{(i+1)}$) is smaller than the current value of SIR$_S$.
\end{proof}
\begin{algorithm}[t]
 	\caption{Distributed position planning for multiple UAVs upon having stochastic interference}\label{alg:fulldistStochastic}
 	\SetKwFunction{Union}{Union}\SetKwFunction{FindCompress}{FindCompress}
 	\SetKwInOut{Input}{input}\SetKwInOut{Output}{output}
 	 	{\footnotesize
 	\Input{Step size $\epsilon$.}
  $i=0$, $d^{(i)}_{N+1}=0$.\\
  $\gamma^{(i)}= E(\textrm{SIR}_{N+1}(\mathbf{d}^{(i)},h))=\frac{p_u\eta^{-1}_{_{\textrm{NLoS}}}}{\left(d^{(i)}_{N+1}\right)^2+h^2}\Upsilon_D$.\label{lineinAlg1Stoc}\\
  Given $d^{(i)}_{N+1}$ and $\gamma^{(i)}$, obtain $d^{(i)}_{N}\rightarrow d^{(i)}_{N-1}\rightarrow \cdots \rightarrow d^{(i)}_{2}$ in order using backward propagation based on~\eqref{eq:num}.\label{lineinAlg3Stoc}\\
Send a message from Tx to $\textrm{\footnotesize UAV}_{1}$ and measure ${\footnotesize \textrm{SIR}_{1}}$.\label{lineinAlg5Stoc}\\
  \uIf{SIR$_{1}<\gamma^{(i)}$}{\label{lineinAlg6Stoc}
  $\gamma^{(i+1)}=\gamma^{(i)}-\epsilon$.\\
  $i=i+1$ and go to line~\ref{lineinAlg1Stoc}.\label{lineinAlg7Stoc}
  }\Else{
  Fix the UAVs at their current positions.
  }
  }
  \end{algorithm}
{\color{black} \begin{remark} For algorithm~\ref{alg:fulldist}, for a given step size $\epsilon$, we have $\textrm{SIR}_S^*-\textrm{SIR}_S^{a}\leq \epsilon$, where SIR$_S^*$ is the optimal solution of the problem and SIR$_S^{a}$ is the achieved SIR$_S$ of the algorithm. The same bound holds for Algorithms~\ref{alg:fulldistStochastic2} and~\ref{alg:fulldistStochastic} by replacing $\textrm{SIR}_S^*$ with $\overline{\textrm{SIR}_S}^*$ and $\textrm{SIR}_S^{a}$ with $\overline{\textrm{SIR}_S}^{a}$ upon having i.i.d interference through x-axis. These bounds can be used as a guideline to tune the step sizes of the algorithms considering the tolerable deviation from the optimal solution. Note that, in general, a smaller step size results in a longer convergence time and a smaller final error.\end{remark}}
\subsection{Discussion on the moment generating function of the interference}\label{sec:betaDistribution}
  \begin{figure}[t]
		\includegraphics[width=.92\linewidth, height=0.45\linewidth]{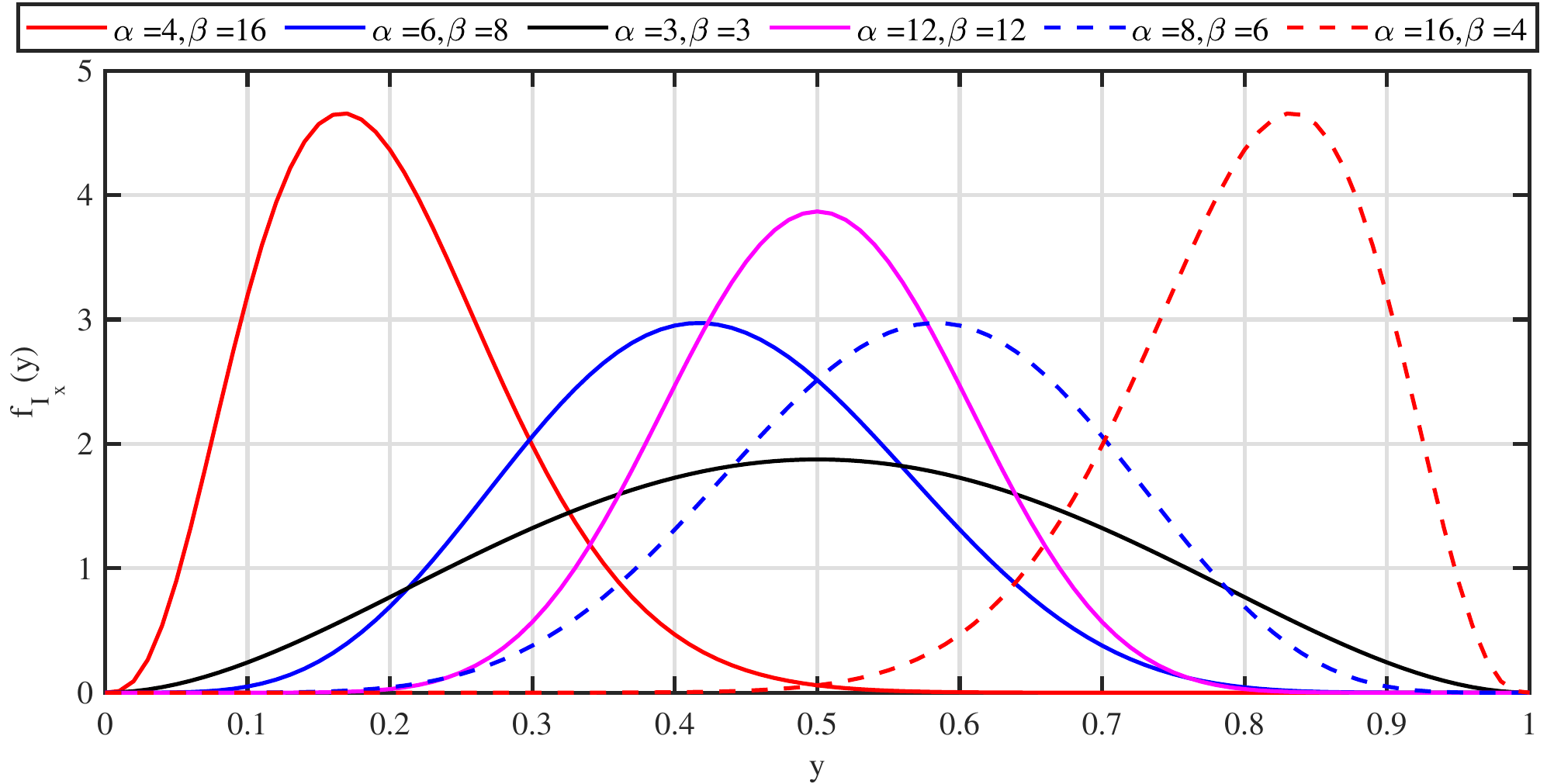}
		\caption{Examples of Beta distribution with different parameters.\label{fig:beta}}
		\vspace{-5mm}
			\end{figure}
			
So far, all the expressions derived in this section can be applied to any given distribution of the interference power by obtaining the moment generating function and its integral at each horizontal point (see \eqref{eq:expectedCalc}). {\color{black}To concertize the results,} in the following, we propose estimating the distribution of the interference power using the Beta distribution, which results in nice closed form expression for the expected value of the SIR expressions.

Any SIR expression given in this section can be written in the following form:
$\textrm{SIR}(x,h)=\frac{S(x,h)}{I_{x}}$, where $S$ is the power of the signal and $I$ is the power of the interference, which is a random variable. {\color{black} Consider $I^{max}$ as an upper bound on the power of interference in the environment: ${I}_{x} <I^{max}$, $\forall x\in[0,D]$}. In this case, the SIR expressions can be written as:
\begin{equation}
    \textrm{SIR}(x,h)=\frac{S(x,h)}{I_{x}}=\frac{S(x,h)/I^{max}}{\Bar{I}_{x}},
\end{equation}
where $\Bar{I}_{x}\triangleq \frac{{I}_{x}}{I^{max}}$ is the normalized interference power. We assume that $I_{x}\neq 0$, $\forall x$, to avoid undefined SIR expressions. Hence, to model $\Bar{I}_{x}$, we look for a family of distribution with support on $[0,1]$, which take the value of zero at $0$ and $1$. One of the good candidates to do such is the Beta distribution, which is known as a general distribution that can be used to approximate the exponential distribution, Rayleigh distribution,  Ricean distribution, and Gamma distribution. In this case, the pdf of $\Bar{I}_{x}$ is given by:
\begin{equation}
    f_{\Bar{I}_{x}}(y) = \frac{y^{\alpha_x-1} (1-y)^{\beta_x-1}}{B(\alpha_x,\beta_x)},
\end{equation}
where $\alpha_x$ and $\beta_x$ are the shaping parameters of the distribution, $B(\alpha_x,\beta_x)\triangleq \frac{\Gamma(\alpha_x) \Gamma (\beta_x)}{\Gamma(\alpha_x+\beta_x)}$, and $\Gamma(.)$ is the gamma function. It can be construed that each point $x\in[0,D]$ is associated with a tuple $(\alpha_x,\beta_x)$. Some examples of the pdf of the Beta distribution considering different shaping parameters are depicted in Fig.~\ref{fig:beta}. 
Using the Beta distribution, for $\alpha_x>1$, we get:
\begin{equation}
\begin{aligned}
E\left(\frac{1}{\Bar{I}_{x}}\right)&= \int_{0}^{1} \frac{1}{y} \frac{y^{\alpha_x-1} (1-y)^{\beta_x-1}}{B(\alpha_x,\beta_x)} dy\\
&= \int_{0}^{1} \frac{y^{\alpha_x-1 -1} (1-y)^{\beta_x-1}}{B(\alpha_x-1,\beta_x)} \frac{B(\alpha_x-1,\beta_x)}{B(\alpha_x,\beta_x)} dy\\
&=\frac{B(\alpha_x-1,\beta_x)}{B(\alpha_x,\beta_x)} \int_{0}^{1} \frac{y^{\alpha_x-1 -1} (1-y)^{\beta_x-1}}{B(\alpha_x-1,\beta_x)} dy \\
&=\frac{B(\alpha_x-1,\beta_x)}{B(\alpha_x,\beta_x)} =\frac{\frac{\Gamma(\alpha_x-1) \Gamma (\beta_x)}{\Gamma(\alpha_x+\beta_x-1)}}{\frac{\Gamma(\alpha_x) \Gamma  (\beta_x)}{\Gamma(\alpha_x+\beta_x)}}\\
&=\frac{\frac{\Gamma(\alpha_x-1) }{\Gamma(\alpha_x+\beta_x-1)}}{\frac{(\alpha_x-1)\Gamma(\alpha_x-1) }{(\alpha_x+\beta_x-1)\Gamma(\alpha_x+\beta_x-1)}} = \frac{\alpha_x+\beta_x-1}{\alpha_x-1},
\end{aligned}
\end{equation}
where the first equality on the last line is the result of the gamma function property $\Gamma(z+1)=z \Gamma (z)$. As compared to \eqref{eq:expectedCalc}, the above result eliminates the need for calculation of the moment generating function of the interference power and its integral at each point. As a result, a new set of nice closed form expressions can be obtained by putting $\Upsilon_{x}=\frac{\alpha_x+\beta_x-1}{\alpha_x-1}$ in \eqref{eq:fixedhStoc}, \eqref{eq:fixedhStoccc}, \eqref{d1Stoc}, \eqref{dkClose}, \eqref{eq:num}. Due to space limitations, we avoid writing the corresponding expressions.

\section{Numerical Results}
   \noindent {\color{black} The simulation parameters are summarized in Table~\ref{table:1}. The general simulation parameters shown in the first sub-table of Table~\ref{table:1} are chosen by following~\cite{chen2018multiple}, also we set $\eta_{_{\textrm{NLoS}}}=\mu_{_{\textrm{LoS}}}$. \begin{table}[h]
  \caption{Simulation parameters. The first sub-table describes the general parameters. For the dual-hop setting, the second and the third sub-tables describe the parameters used in Fig.~\ref{fig:locus} and Fig.~\ref{fig:rate}, respectively. For the multi-hop setting,  the fourth and the fifth sub-tables describe the parameters used in Figs.~\ref{fig:minpos},\ref{fig:minpower} and Figs.~\ref{fig:conv},\ref{fig:gainDist}, respectively.}\label{table:1}
\centering \resizebox{\columnwidth}{!}{
{\small
 \begin{tabular}{|c|c| c| c| c|c|c|c|} 
 \hline
 \multirow{2}{3.5em}{General:}& $f_c$ & $C_{_{\textrm{LoS}}}$& $C_{_{\textrm{NLoS}}} $& -& -& -& - \\ 
  \cline{2-8}
 &$2\textrm{GHz}$ & $10^{0.01}$ & $10^{2.1}$ &
 -& -& -& - \\ 
 \hline\hline
 \multirow{2}{2.5em}{Fig.~\ref{fig:locus}:}
  & $p_u$ & $D $& -& -& -& -& -\\
  \cline{2-8}
 &$1\textrm{W}$ & $35\textrm{m}$&-& -& -& -& -  \\
  \hline\hline
 \multirow{2}{2.5em}{Fig.~\ref{fig:rate}:}&
  $p_{_{\textrm{MSI}}}$& $Y_{_{MSI}}$&$X_{_{MSI}}$ &$D$& -& -& -\\
  \cline{2-8}
 &$20\textrm{W}$& $30\textrm{m}$ &$30\textrm{m}$& $35\textrm{m}$& - & - & -\\
  \hline\hline
   \multirow{2}{3em}{Figs.~\ref{fig:minpos},\ref{fig:minpower}:}&
  $p_t$ &$p_{_{\textrm{MSI}}}$& $D$ & $h$& -& -& -\\
  \cline{2-8}
 & $80\textrm{W}$ & $80\textrm{W}$ & $1000\textrm{m}$ & $20\textrm{m}$ & -& -& - \\
  \hline\hline
     \multirow{2}{3em}{Figs.~\ref{fig:conv},\ref{fig:gainDist}:}&
   $p_u$& $p_t$& $p_{_{\textrm{MSI}}}$& $h$&$D$ &$X_{_{\textrm{MSI}}}$&$Y_{_{\textrm{MSI}}}$\\
  \cline{2-8}
 & $1\textrm{W}$ & $80\textrm{W}$ & 80\textrm{W} & $20\textrm{m}$ & $1000\textrm{m}$& $500\textrm{m}$& $400\textrm{m}$\\
  \hline
 \end{tabular}
 }
   \vspace{-5mm}
   }
\end{table}

  \subsection{Dual-Hop Setting }
  Considering the transmission power of the UAV and the distance between the Tx and Rx as given in the second sub-table of Table~\ref{table:1},  Fig.~\ref{fig:locus} depicts the locus described in Lemma~\ref{th:main} for various parameters.}

  Considering the solid black line and the dotted red line as the references, as expected, increasing $p_t$ (the marked blue and dashed magenta lines) shifts the locus toward the Rx. Also, considering the dotted red line and the dashed magenta line as the references, bringing the MSI closer to the Tx/Rx (the solid black and marked blue lines) shifts the locus downward, which is equivalent to a decrease in the required UAV altitude. 
		 \begin{figure}[t]
		\minipage{1\linewidth}
		\includegraphics[width=3.in, height=1.74in]{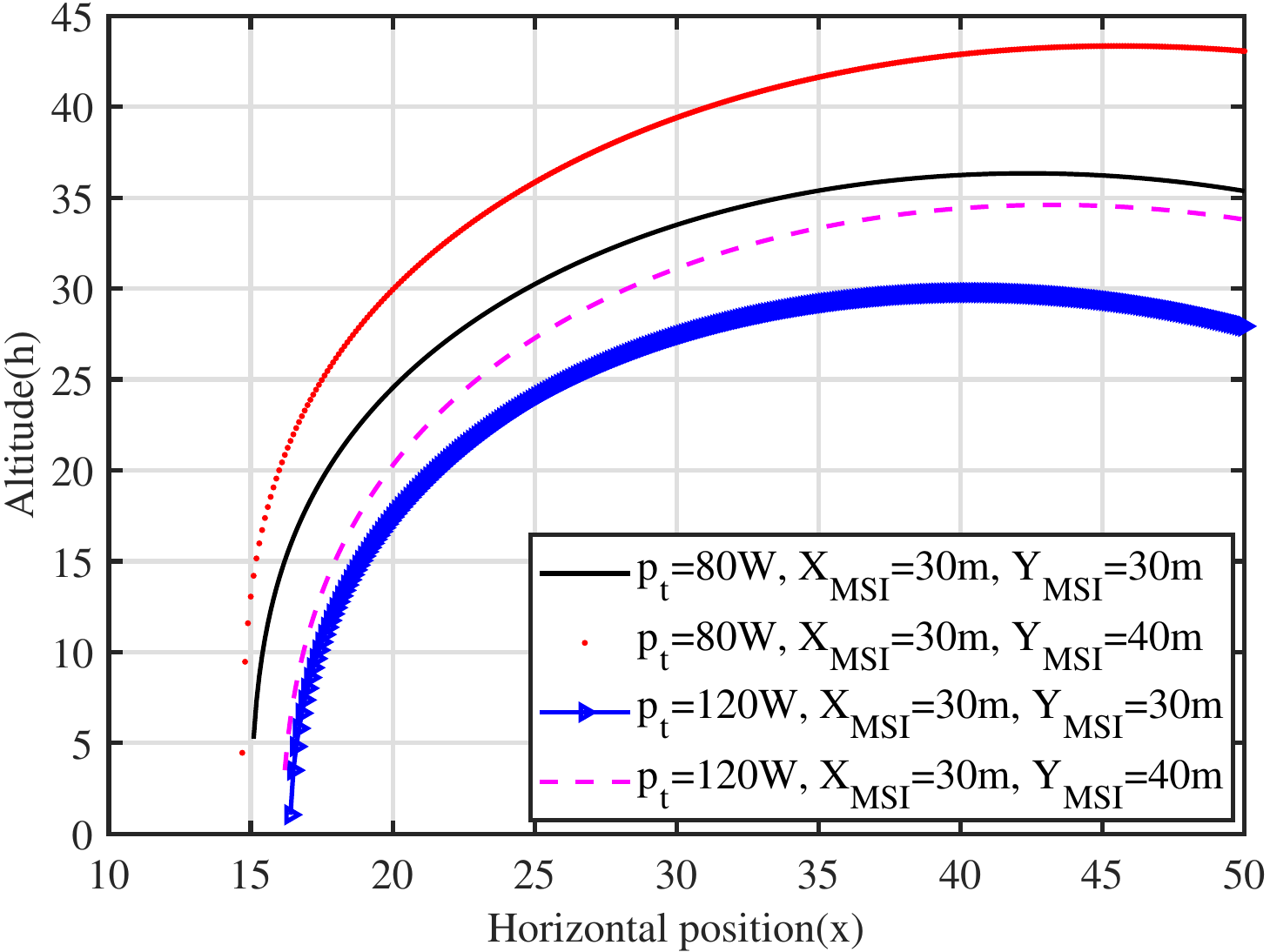}
		\caption{The locus of the points described in Lemma~\ref{th:main} for different parameters. \label{fig:locus}}
		\endminipage
		
		\minipage{1\linewidth}
		\includegraphics[width=3.in, height=1.74in]{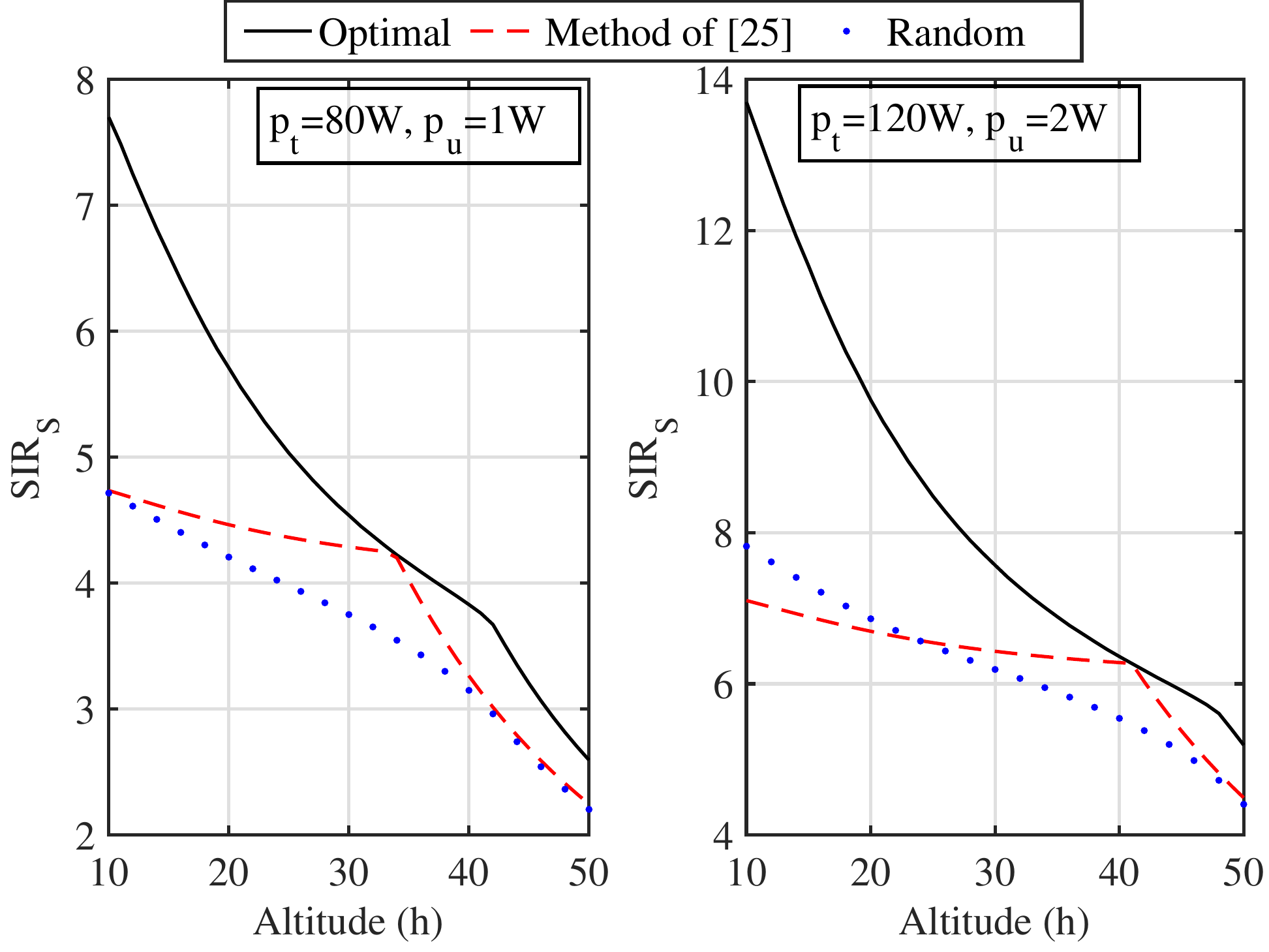}
		\caption{Comparison between the $\textrm{SIR}_S$ obtained using our optimal approach, the method in~\cite{chen2018multiple}, and the random placement for different parameters considering a single UAV.\label{fig:rate}}
		\endminipage
		
		\minipage{1\linewidth}
		\includegraphics[width=3.in, height=1.74in]{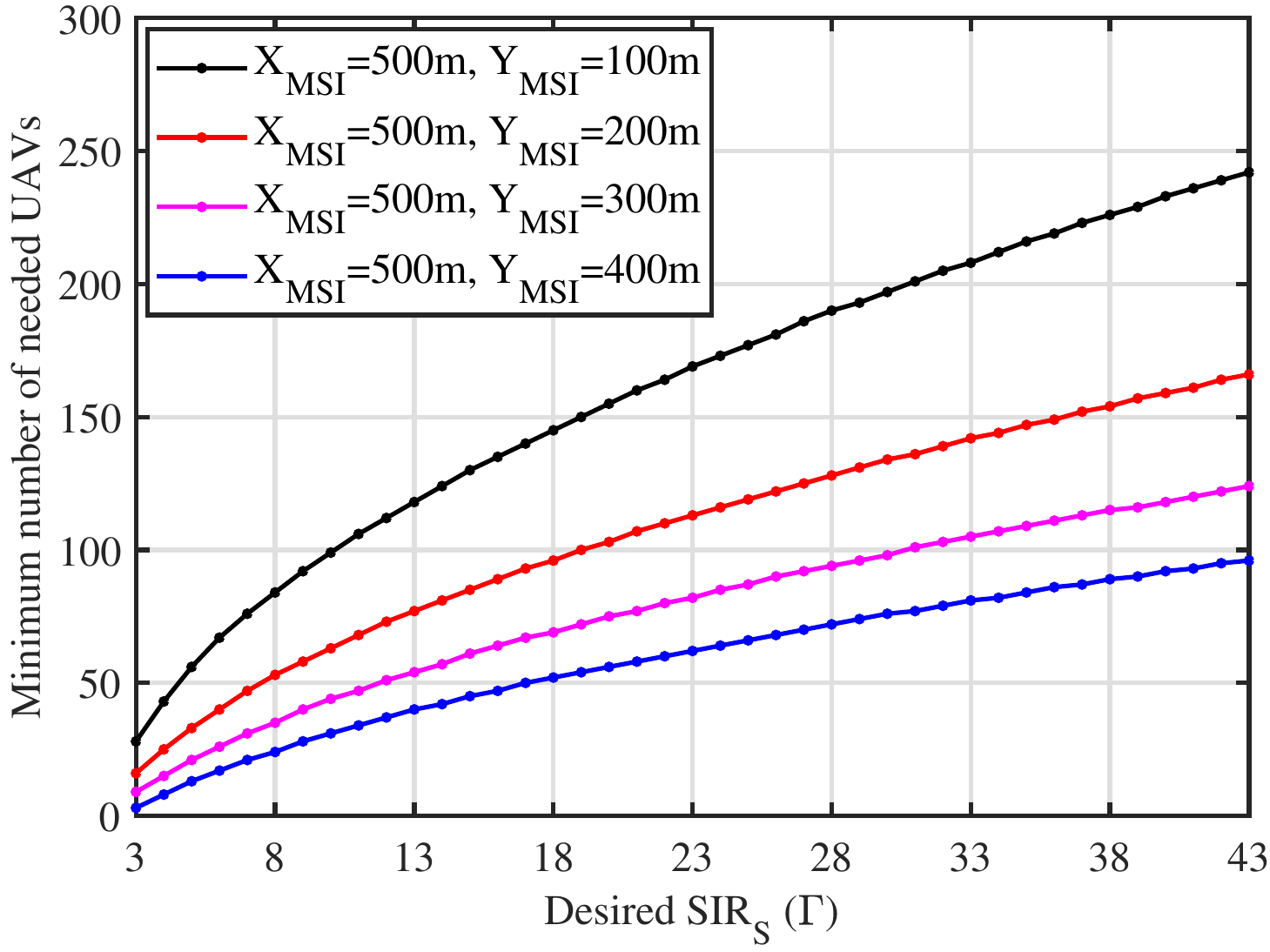}
		\caption{The minimum required number of UAVs to satisfy various values of $\textrm{SIR}_S$ for different positions of the MSI. \label{fig:minpos}}
		\endminipage
		
		\minipage{1\linewidth}
		\includegraphics[width=3.in, height=1.74in]{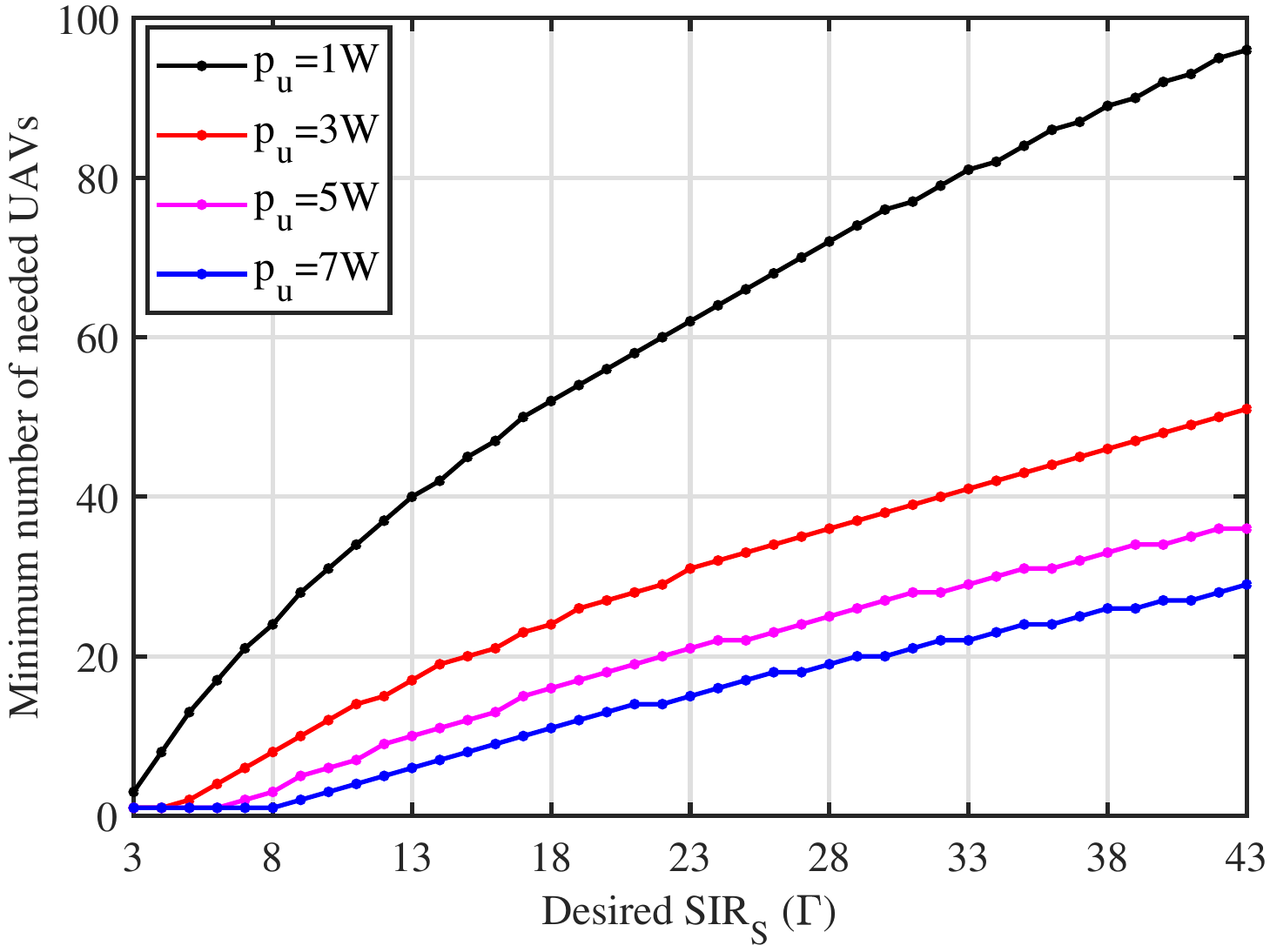}
		\caption{The minimum required number of UAVs to satisfy various values of $\textrm{SIR}_S$ for different UAV transmission powers.\label{fig:minpower}}
		\endminipage
		\vspace{-6.5mm}
	\end{figure}
		
Fig.~\ref{fig:rate} compares the $\textrm{SIR}$ of the system obtained using Theorem~\ref{th:givenh} to both the random placement, the performance of which is obtained by randomly placing the UAV in $1000$ Monte-Carlo iterations, and the method described in~\cite{chen2018multiple}, which does not capture the existence of the MSI (see Section~\ref{sec:intro}). {\color{black} In this simulation, it is assumed that the altitude of the UAV changes in the interval $[10\textrm{m},50\textrm{m}]$ and the simulation is conducted for two realizations of $p_t$ and $p_u$ described in the figure.  The rest of the parameters are described in the third sub-table of Table~\ref{table:1}}. As can be seen, the difference between the performance of our approach and the baselines is more prominent in low altitudes (up to $65\%$ increase in $\textrm{SIR}_S$). Also, on average, our method leads to $30.14\%$ and $25.73\%$ increase in $\textrm{SIR}_S$ as compared to the random placement and the method of~\cite{chen2018multiple}, respectively.


	

  \subsection{{\color{black} Multi-Hop Setting }}
  {\color{black} Considering the parameters described in the fourth sub-table of Table~\ref{table:1},} based on~\eqref{eq:Opt}, the minimum required number of UAVs to satisfy different values of $\textrm{SIR}_S$ when $p_u=1\textrm{W}$ for multiple MSI positions  (when $X_{_{\textrm{MSI}}}=500\textrm{m},Y_{_{\textrm{MSI}}}=400\textrm{m}$ for multiple $p_u$ values) is depicted in Fig.~\ref{fig:minpos} (Fig.~\ref{fig:minpower}). From Fig.~\ref{fig:minpos}, it can be observed that as the MSI gets closer to the Tx/Rx the required number of UAVs increases. Also, from Fig.~\ref{fig:minpower}, it can be seen that by increasing the $p_u$ the required number of UAVs decreases. 
   
   {\color{black} Considering the parameters described in the fifth sub-table of Table~\ref{table:1}}, $d_{min}=4\textrm{m}$, and $\epsilon=0.1$, Fig.~\ref{fig:conv} depicts the performance of our distributed algorithm  (Algorithm~\ref{alg:fulldist}) for various number of UAVs in the network. From Fig.~\ref{fig:conv}, it can be seen that as the number of UAVs increases our algorithm achieves larger values of $\textrm{SIR}_S$ with a faster convergence speed. The faster convergence is due to the larger coverage length when having a large number of UAVs. Fig.~\ref{fig:gainDist} reveals the significant increase in the $\textrm{SIR}_S$ obtained through comparing the achieved $\textrm{SIR}_S$ using our distributed algorithm with both the method described in~\cite{chen2018multiple} and the random placement, the performance of which is obtained by randomly placing the UAVs in $1000$ Monte-Carlo iterations.

			 \begin{figure}[t]
		
		\minipage{1\linewidth}
		\includegraphics[width=3.in, height=1.74in]{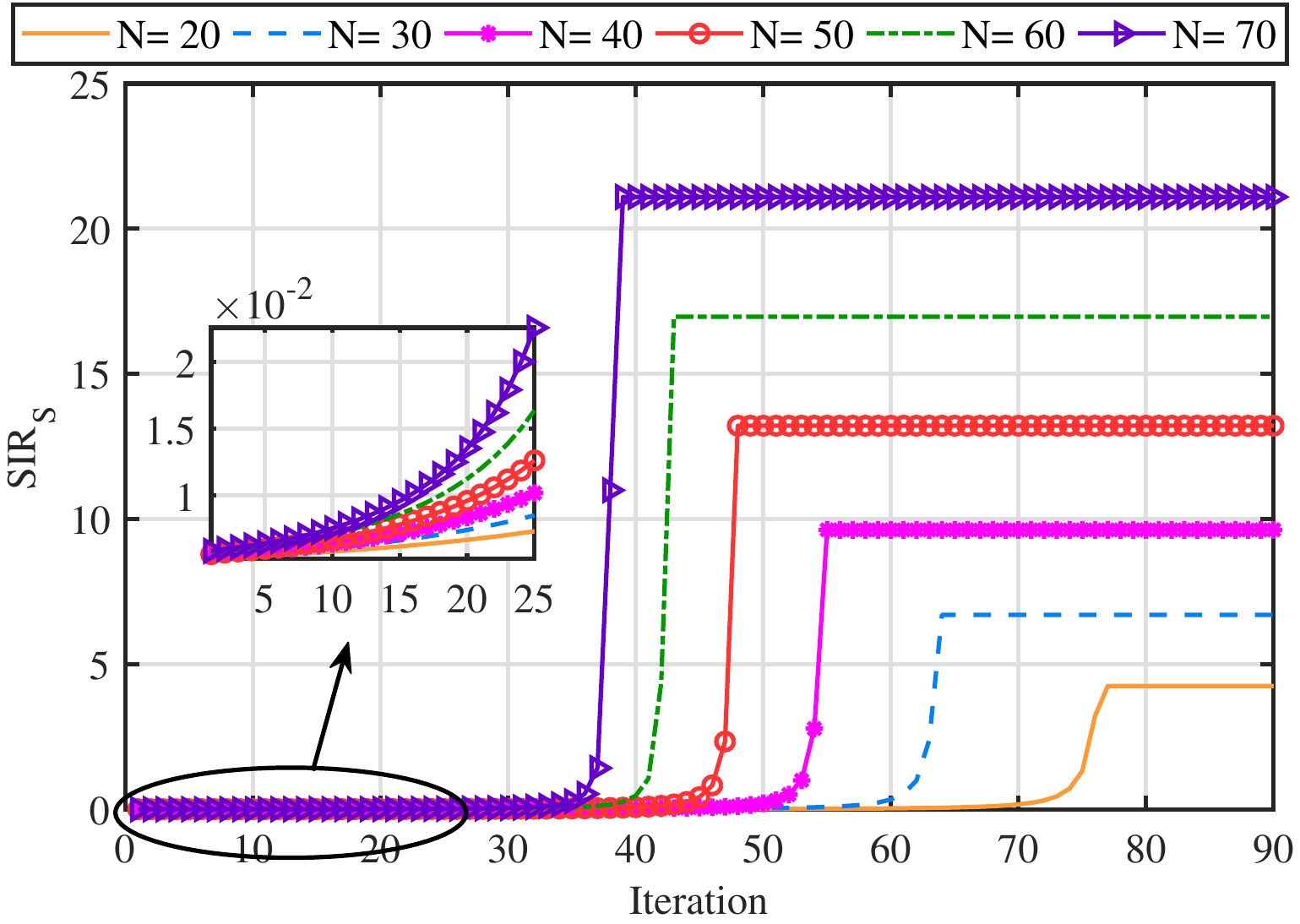}
		\caption{The $\textrm{SIR}_S$ w.r.t the iteration for our proposed distributed algorithm considering various number of UAVs in the system. \label{fig:conv}}
		\endminipage
		
		\minipage{1\linewidth}
		\includegraphics[width=3.in, height=1.74in]{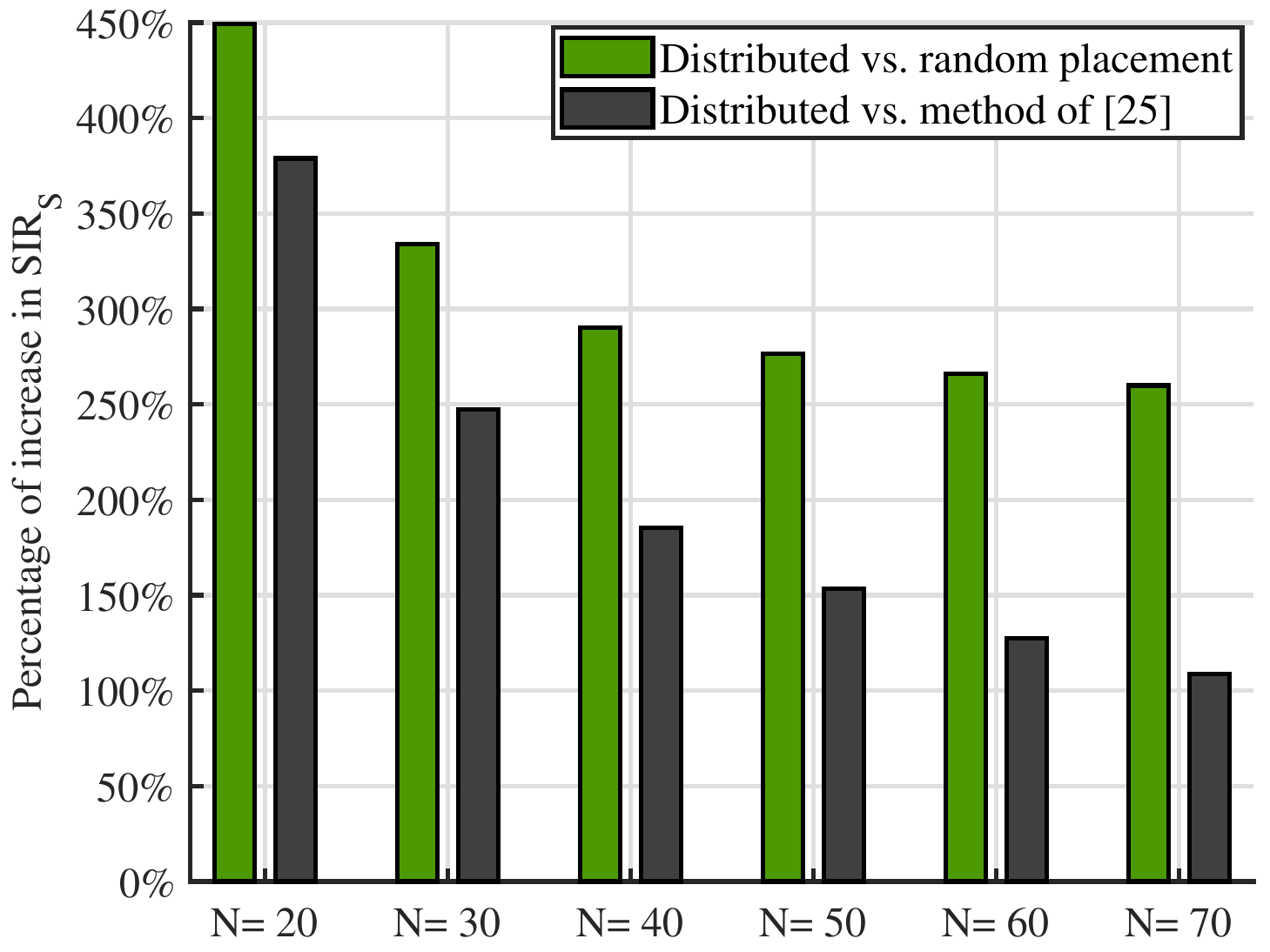}		\caption{Percentage of increase in $\textrm{SIR}_S$ by comparing our distributed algorithm with the random placement and the method in~\cite{chen2018multiple} for various number of UAVs ($N$).\label{fig:gainDist}}
		\endminipage
		
		\minipage{1\linewidth}
		\includegraphics[width=3.in, height=1.74in]{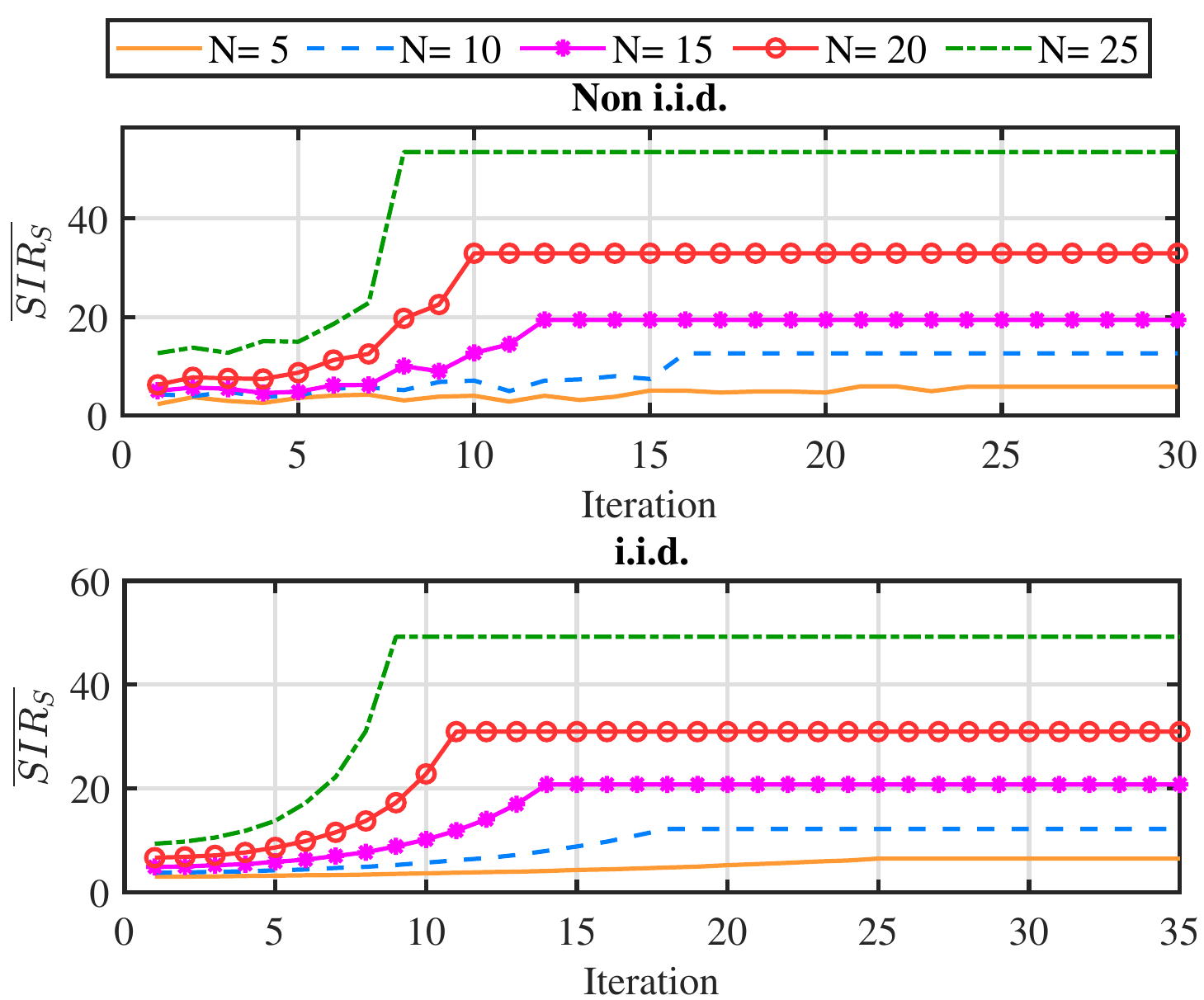}
		\caption{The $\textrm{SIR}_S$ w.r.t the iteration for our proposed distributed algorithm (Algorithm~\ref{alg:fulldistStochastic}) upon existence of stochastic interference, which is non-i.i.d (top subplot) and i.i.d. (bottom subplot) in through the $x$-axis. The number of UAVs in the system is denoted by $N$. \label{fig:conv2}}
		\endminipage
				\vspace{-6.5mm}
	\end{figure}

			Position planning upon having stochastic interference leads to similar results as compared to the previous performed simulations, especially upon having i.i.d. interference along $x$-axis. As an example, considering the parameters of the previous simulation, Fig.~\ref{fig:conv2} depicts the performance of our distributed algorithm (Algorithm~\ref{alg:fulldistStochastic}) for various number of UAVs in the network. We assume that the interference follows the Beta distribution (see Section~\ref{sec:betaDistribution}). We first consider the interference to be non-i.i.d. through the $x$-axis, where $\alpha_x$ and $\beta_x$ are chosen such that $\frac{\alpha_x+\beta_x}{\alpha_x}$ ranges from $0.6 \Lambda^{-1}$ to $\Lambda^{-1}$, $\forall x\in [0,D]$, where $\Lambda = \left(\eta_{_{\textrm{NLoS}}}h^2\right)^{-1}$ is the power of the signal received at the UAV if the UAV is located above the Tx (at $x=0$) for a normalized transmitter power. Also, the performance of the algorithm upon having i.i.d. interference with $\frac{\alpha_x+\beta_x}{\alpha_x}=0.8 \Lambda$, $\forall x\in [0,D]$, is depicted in the bottom subplot of Fig.~\ref{fig:conv2}. Comparing the two subplots, upon having the non-i.i.d. interference, moving from one iteration to the next may not lead to a better SIR\_S due to the unpredictable power of interference at different positions during the iterations (see \eqref{eq:num}); however, in both cases, the convergence is achieved through a few iterations. Also, comparing the performance of our algorithm with the baselines will result in similar results depicted in Fig.~\ref{fig:gainDist}, which is omitted to avoid redundancy.

		{\color{black}
		\subsection{Adjusting the altitudes and the horizontal positions of the UAVs }
				We first conduct a set of simulations to demonstrate the effect of the altitudes of the UAVs on the performance of Algorithm~\ref{alg:fulldist}, when the UAVs have the same altitude. The results are depicted in Fig.~\ref{fig:diff_alt} assuming the number of UAVs $N=50$.
		We use the same simulation parameters as in Figs.~\ref{fig:conv} and \ref{fig:gainDist} except the value for $Y_{_{\textrm{MSI}}}$, which is variable. As can be seen, as the MSI gets closer to the UAVs (smaller values of $Y_{_{\textrm{MSI}}}$), the altitude that results in the best performance increases. However, increasing the altitude after a certain value results in performance reduction. This is because increasing the altitude initially results in decreasing the effect of interference on the UAVs; however, after a certain point the link from the Tx to the first UAV and from the last UAV to the Rx become the bottleneck and increasing the altitude leads to a reduction in the SIR$_{S}$.
		
		To demonstrate the effect of adjustments of the altitudes and the horizontal positions of the UAVs on the SIR$_S$ using Algorithm~\ref{alg:obtainHinDist}, we fix the location of MSI at $X_{_{\textrm{MSI}}}=500\textrm{m},Y_{_{\textrm{MSI}}}=150\textrm{m}$ and choose the corresponding altitude that results in the best performance of Algorithm~\ref{alg:fulldist} from the red curve with circle markers in Fig.~\ref{fig:diff_alt}, i.e., $h=220\textrm{m}$, as the initial altitude of all the UAVs. Also, the initial horizontal positions of the UAVs are obtained using~Algorithm~\ref{alg:fulldist}. Fig.~\ref{fig:config_UAVs} depicts the location of the Tx, the MSI, the Rx, and the UAVs in the 2-D plane for different number of iterations of Algorithm~\ref{alg:obtainHinDist}. Two key observations from  Fig.~\ref{fig:config_UAVs} are as follows: i) the algorithm brings the first few UAVs closer to the Tx, the last few UAVs closer to the Rx, while moving the middle UAVs away from the MSI; ii) the algorithm leads to larger horizontal separations between the first few and the last few UAVs, while the middle UAVs are closer. The reason behind these observations is that the effect of interference on the middle UAVs is larger since they are closer to the MSI; and thus the algorithm brings them closer to each other and moves them away from the MSI. Note that these observations and the semicircle 2-D shape of the UAVs configuration in Fig.~\ref{fig:config_UAVs} are parameter specific and may vary for a different parameter setting, e.g., a different location for the MSI. Fig.~\ref{fig:SIR_greedy} depicts the increase in SIR$_S$ achieved using Algorithm~\ref{alg:obtainHinDist}, where $25\%$ increase in SIR$_S$ can be observed.
		}
		
			 \begin{figure}[t]
\centering
\minipage{1\linewidth}
		\includegraphics[width=3.in, height=1.74in]{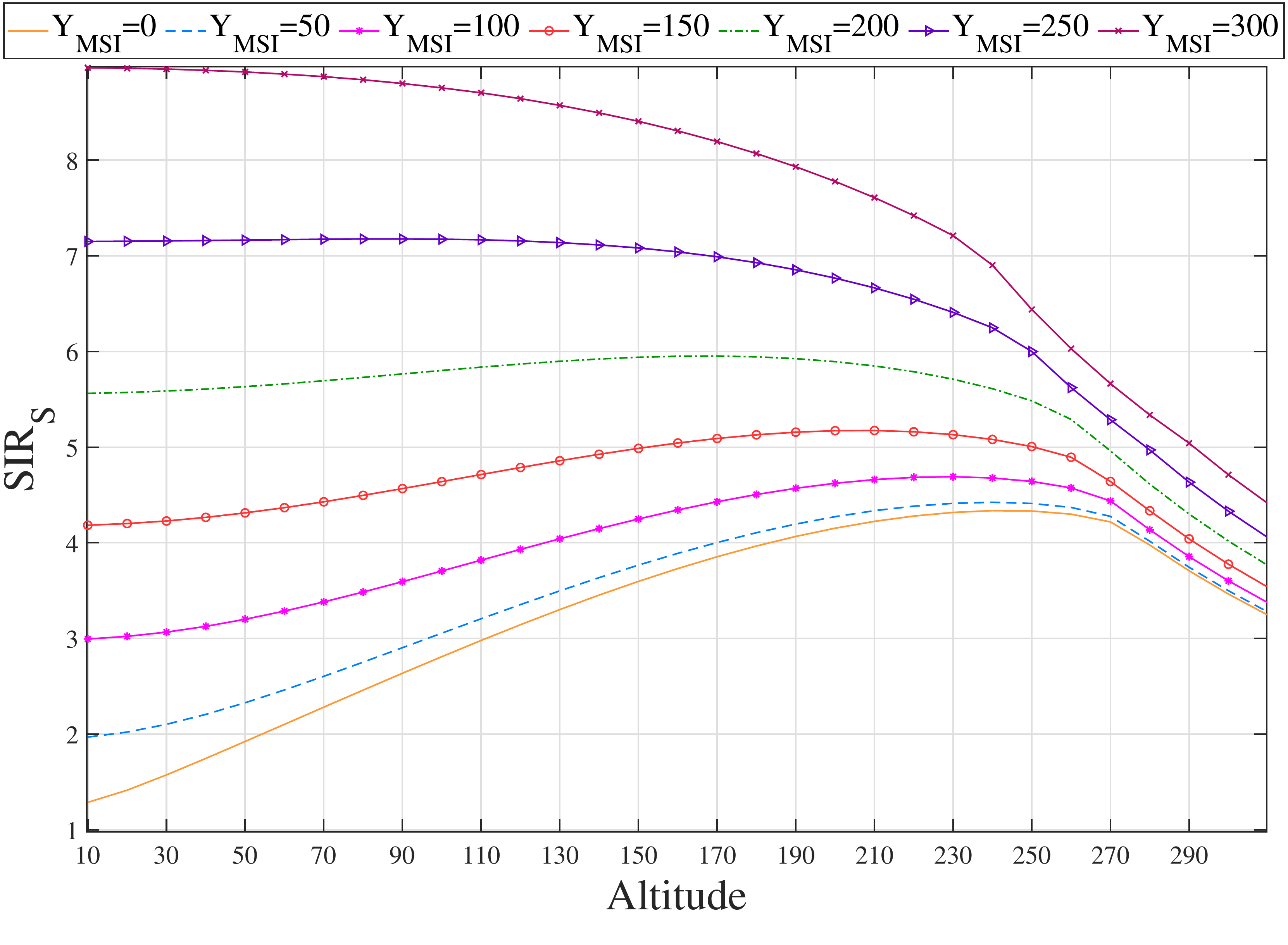}
		\caption{{\color{black}The final SIR$_S$ achieved through Algorithm~\ref{alg:fulldist}, for different altitudes, when all the UAVs have the same altitude specified in the x-axis.\label{fig:diff_alt}}}
		\endminipage

\minipage{1\linewidth}
		\includegraphics[width=3.in, height=1.74in]{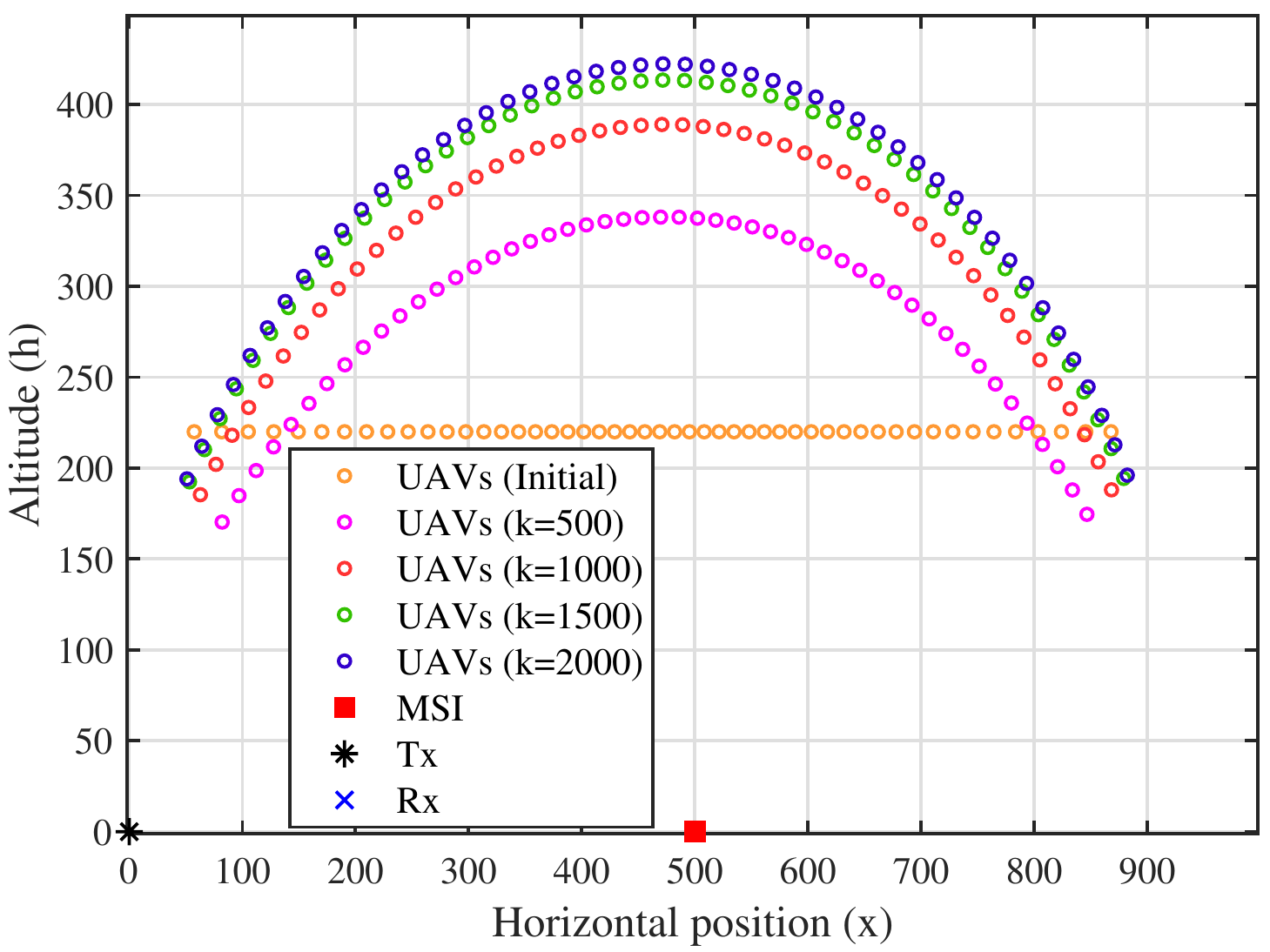}
		\caption{{\color{black}Locations of the UAVs in x-h plane for different number of iterations (k) of Algorithm~\ref{alg:obtainHinDist}. Positions of the Tx, the MSI, and the Rx are also depicted in the figure.\label{fig:config_UAVs}}}
			\endminipage

\minipage{1\linewidth}
		\includegraphics[width=3.in, height=1.74in]{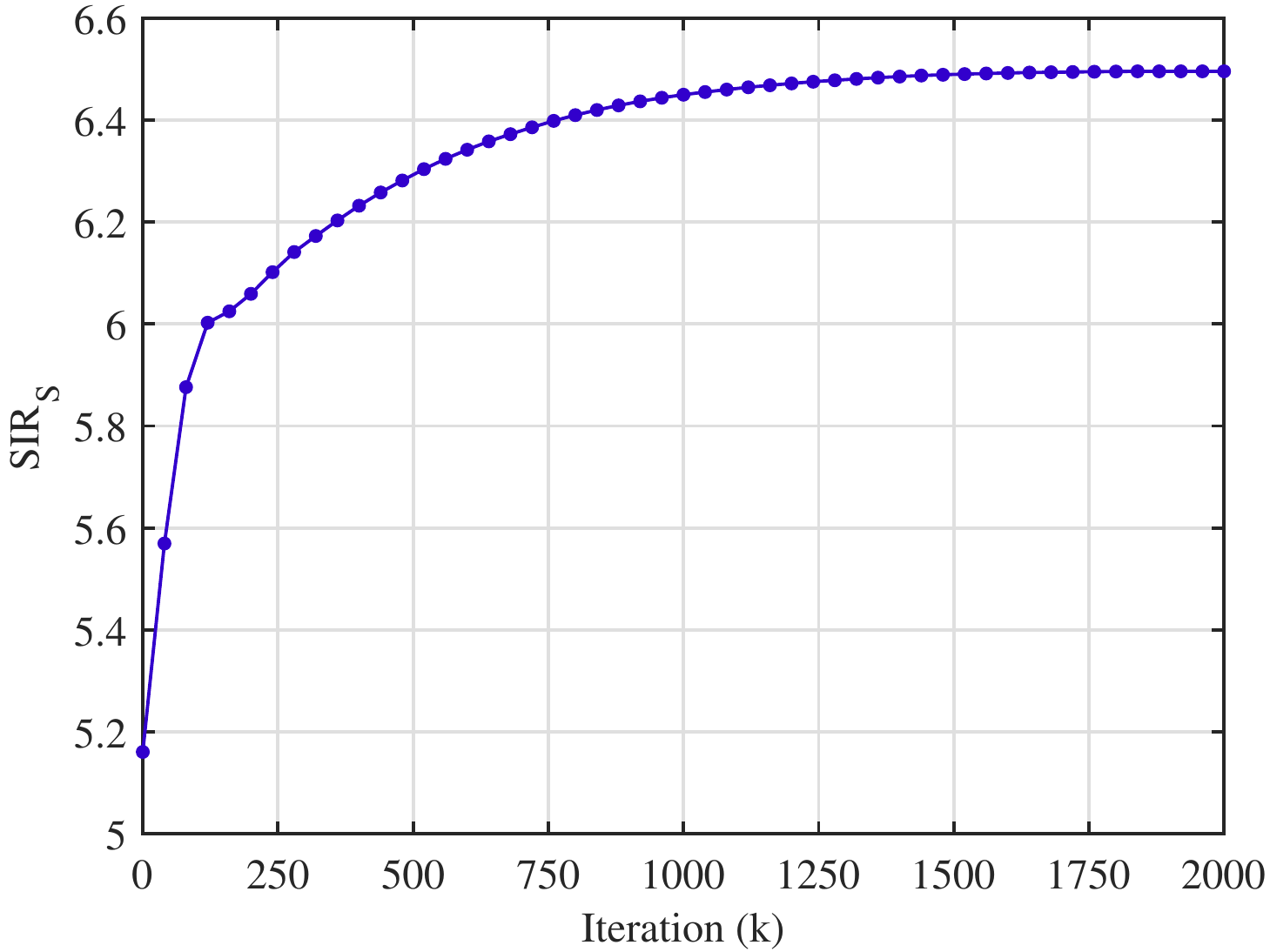}		\caption{{\color{black}The SIR$_S$ w.r.t the iteration number (k) of Algorithm~\ref{alg:obtainHinDist}. \label{fig:SIR_greedy}}}
		\endminipage
		\end{figure}
		
			\section{Conclusion}
\noindent
In this work, we studied the UAV-assisted relay wireless communication paradigm considering the presence of interference in the environment. We investigated the UAV(s) position planning considering two scenarios: i) existence of an MSI, and ii) existence of stochastic interference. For each scenario, we first endeavored to maximize the (average) SIR of the system considering a single UAV in the network. Afterward, for each scenario, we studied the position planning in multi-hop relay scheme, in which the utilization of multiple UAVs is feasible. To this end, we first proposed a theoretical approach, which simultaneously determines the minimum number of needed UAVs and their optimal positions so as to satisfy a desired (average) SIR of the system. Second, for a given number of UAVs in the network, we proposed a distributed algorithm along with its performance guarantee, which solely requires message exchange between the adjacent UAVs so as to maximize the (average) SIR of the system. Furthermore, we illustrated the performance of our methods through numerical simulations. The methodology of this work can inspire multiple future works revisiting the previously studied problems in the context of UAV-assisted relay wireless communications considering the existence of interference in the environment. {\color{black} Also, investigating the studied problems in this paper upon having multiple pairs of Tx and Rx is a promising future work.}

\bibliographystyle{IEEEtran}
\bibliography{ABSbib}
\end{document}